\documentclass[11pt]{article}
\usepackage[letterpaper,margin=1in]{geometry}
\usepackage{amsmath,amssymb,amsthm,mathtools}
\usepackage{aliascnt}
\usepackage{booktabs,array,enumitem,microtype,needspace}
\usepackage[dvipsnames]{xcolor}
\usepackage[hidelinks]{hyperref}
\usepackage[nameinlink,noabbrev,capitalise]{cleveref}
\usepackage{algorithm}
\usepackage[noend]{algpseudocode}
\algrenewcommand{\textproc}[1]{\text{\emph{#1}}}

\newtheorem{theorem}{Theorem}[section]
\newaliascnt{lemma}{theorem}
\newtheorem{lemma}[lemma]{Lemma}
\aliascntresetthe{lemma}
\newaliascnt{corollary}{theorem}
\newtheorem{corollary}[corollary]{Corollary}
\aliascntresetthe{corollary}
\newaliascnt{proposition}{theorem}

\aliascntresetthe{proposition}
\newaliascnt{claim}{theorem}

\aliascntresetthe{claim}
\theoremstyle{definition}
\newaliascnt{definition}{theorem}
\newtheorem{definition}[definition]{Definition}
\aliascntresetthe{definition}
\newaliascnt{remark}{theorem}

\aliascntresetthe{remark}
\crefname{theorem}{Theorem}{Theorems}
\Crefname{theorem}{Theorem}{Theorems}
\crefname{lemma}{Lemma}{Lemmas}
\Crefname{lemma}{Lemma}{Lemmas}
\crefname{corollary}{Corollary}{Corollaries}
\Crefname{corollary}{Corollary}{Corollaries}
\crefname{proposition}{Proposition}{Propositions}
\Crefname{proposition}{Proposition}{Propositions}
\crefname{claim}{Claim}{Claims}
\Crefname{claim}{Claim}{Claims}
\crefname{definition}{Definition}{Definitions}
\Crefname{definition}{Definition}{Definitions}
\crefname{remark}{Remark}{Remarks}
\Crefname{remark}{Remark}{Remarks}

\newcommand{\ED}{\operatorname{WED}}
\newcommand{\SD}{\ED^{>H}}
\newcommand{\gap}{\bot}
\newcommand{\eps}{\varepsilon}
\newcommand{\softO}{\widetilde O}
\newcommand{\R}{\mathbb R}

\newcommand{\Prb}{\mathop{\rm Pr}}
\newcommand{\E}{\mathbb E}
\newcommand{\cX}{\mathcal X}
\newcommand{\cY}{\mathcal Y}
\newcommand{\cM}{\mathcal M}

\newif\ifshowcomments
\showcommentstrue

\title{A Strongly Subquadratic $(3+\varepsilon)$-Approximation
for\\ Weighted Edit Distance over Arbitrary Metrics}
\author{%
  Ethan Mader\thanks{Department of Computer Science, Purdue University, West Lafayette, Indiana.
    \href{mailto:emader@purdue.edu,btavasol@purdue.edu}{\nolinkurl{{emader, btavasol}@purdue.edu}}.}
  \and
  Borna Tavasoli\footnotemark[1]
  \and
  Jihan Wang\thanks{Department of Computer Science, The University of Illinois Urbana-Champaign, Urbana, Illinois.\newline%
    \hspace*{1.8em}\href{mailto:jihanw2@illinois.edu}{\mbox{\nolinkurl{jihanw2@illinois.edu}}}.}
}
\date{September 13, 2026}

\begin{document}
\hypersetup{pageanchor=false}
\maketitle

\begin{abstract}

We study weighted edit distance between two strings of total length \(n\),
where edit costs are induced by an arbitrary metric.  For equal-length inputs,
Kuszmaul~\cite{Kuszmaul} gave an \(O(n^\delta)\)-approximation with
\(\softO(n^{2-\delta})\) running time for every fixed \(0<\delta<1\).  We
give the first constant-factor approximation for weighted edit distance over
arbitrary metrics in strongly subquadratic running time.  For every
\(0<\varepsilon\le1\), our randomized algorithm runs in
\(\softO(n^{7/4}/\varepsilon^8)\) time and returns a
\((3+\varepsilon)\)-approximation with probability at least \(1-n^{-10}\).  The algorithm
allows unequal input lengths and places no bound on the ratio between
edit costs.  

\end{abstract}

\thispagestyle{empty} 
\clearpage 

\setcounter{page}{1}
\hypersetup{pageanchor=true}
\section{Introduction}

The classical \emph{unit-cost edit distance} between two strings is the
minimum number of insertions, deletions, and substitutions needed to
transform one into the other.  Introduced by Levenshtein in the study of
error-correcting codes~\cite{Levenshtein}, it is a fundamental measure of
dissimilarity between strings. It is widely used in sequence alignment, text
comparison, document processing, and approximate pattern
matching.  In many of these applications, however, not all edits are equally costly, and the natural measure charges each operation according to the symbols it involves.

Formally, let \(X,Y\in\Sigma^*\) be two strings over an alphabet \(\Sigma\), and set
\(n=|X|+|Y|\).  The input includes a metric \(\rho\) on
\(\Sigma\cup\{\gap\}\), where \(\gap\notin\Sigma\) is a gap symbol.
Deleting or inserting a symbol \(a\in\Sigma\) costs \(\rho(a,\gap)\).
Substituting \(a\) by \(b\in\Sigma\) costs \(\rho(a,b)\).  The \emph{weighted
edit distance} \(\ED_\rho(X,Y)\) is the minimum total cost of a sequence of
these operations that transforms \(X\) into \(Y\).  We study this problem
for arbitrary metrics.  Costs need not be integral or normalized, and the
ratio between the largest and smallest positive costs may be arbitrarily large.
We abbreviate \(\ED_\rho\) to \(\ED\) when the metric is fixed.

\paragraph{Exact computation.}
The classical dynamic programming (DP) algorithm of Wagner and Fischer computes
edit distance in
\(O(n^2)\) time~\cite{WagnerFischer}.  Masek and Paterson obtained an
\(O(n^2/\log n)\)-time algorithm for a fixed finite alphabet and fixed
rational edit costs~\cite{MasekPaterson}.  Ukkonen~\cite{Ukkonen}
gave an \(O(nk)\)-time exact algorithm when the unit-cost
edit distance is \(k\).  Landau and Vishkin
subsequently obtained an \(O(n+k^2)\)-time exact
algorithm~\cite{LandauVishkin}.  These improvements do not yield a strongly
subquadratic time in the worst case.  Assuming the Strong
Exponential Time Hypothesis (SETH)~\cite{ImpagliazzoPaturi}, Backurs and Indyk
proved that even
unit-cost edit distance admits no
\(O(n^{2-\delta})\)-time exact algorithm for any constant
\(\delta>0\)~\cite{BackursIndyk}.  This barrier motivates the search for
faster approximation algorithms.

\begingroup
\setlength{\emergencystretch}{1em}
\paragraph{Unit-cost approximation.}
Bar-Yossef, Jayram, Krauthgamer, and Kumar gave an
\(n^{3/7}\)-approximation in \(\widetilde O(n)\)
time~\cite{BarYossefEtAl}.
Batu, Erg\"un, and Sahinalp improved this to an
\(n^{1/3+o(1)}\)-approximation in \(\widetilde O(n)\)
time~\cite{BatuErgunSahinalp}.  Andoni and Onak then obtained a
\(2^{O(\sqrt{\log n\log\log n})}\)-approximation in
\(O(n^{1+o(1)})\)
time~\cite{AndoniOnak}, and Andoni, Krauthgamer, and Onak subsequently
achieved a
\((\log n)^{O(1/\delta)}\)-approximation in
\(O(n^{1+\delta})\) time for every fixed
\(\delta>0\)~\cite{AndoniKrauthgamerOnak}.

Chakraborty, Das, Goldenberg, Kouck\'y, and Saks gave the first classical constant-factor approximation in
\(\widetilde O(n^{12/7})\) time~\cite{CDGKS}.  Goldenberg, Rubinstein, and
Saha then obtained a randomized
\((3+o(1))\)-approximation in \(n^{8/5+o(1)}\)
time~\cite{GoldenbergRubinsteinSaha}.
For inputs whose edit distance is nearly linear in \(n\), Kouck\'y and
Saks~\cite{KouckySaks} and Brakensiek and Rubinstein~\cite{BrakensiekRubinstein}
gave constant-factor approximations in near-linear time.  Andoni and Nosatzki
then gave a randomized constant-factor approximation for all inputs in
\(O(n^{1+\delta})\) time for every fixed
\(\delta>0\), with an approximation factor depending on
\(\delta\)~\cite{AndoniNosatzki}.  More recently, Mao and Rubinstein
obtained a randomized
\((1+\varepsilon)\)-approximation in quasi-strongly subquadratic time
for every constant \(\varepsilon>0\)~\cite{MaoRubinstein}.  
Fox, Mader, Quanrud, Tavasoli, and Wang achieved a randomized
\((3+\varepsilon)\)-approximation in \(\widetilde O(n^{6/5}k^{2/5}\operatorname{poly}(1/\varepsilon))\) time~\cite{FMQTW26}, where \(k\) is the edit distance.
\par
\endgroup

\paragraph{Bounded weighted edit distance.}
Recent work has made substantial progress on bounded weighted edit distance.
The weights are normalized so that matching identical symbols costs zero and
every other edit operation costs at least one.  The running times depend on
the weighted distance \(k\) under this normalization.  Das, Gilbert, Hajiaghayi,
Kociumaka, and Saha gave the first
exact \(O(n+\operatorname{poly}(k))\)-time algorithm for weighted string edit
distance~\cite{DasEtAl}.  Cassis, Kociumaka,
and Wellnitz improved this running time~\cite{CassisKociumakaWellnitz}.
For \(\sqrt n\le k\le n\), they also showed that their running time is optimal
up to subpolynomial factors under the All-Pairs Shortest Paths
(APSP) hypothesis ~\cite{VassilevskaWilliamsWilliams}.  For integer weights
bounded by \(W\), Gorbachev and Kociumaka subsequently obtained an
\(\widetilde O(n+Wk^2)\)-time algorithm~\cite{GorbachevKociumaka}.
These results exploit small normalized distance.

\paragraph{Arbitrary metrics.}
The same quadratic-time DP computes both unit-cost edit distance and weighted
edit distance over arbitrary metrics exactly.  The approximation guarantees
for arbitrary metric costs have nevertheless remained much weaker than those
for unit-cost edit distance.  Prior to this work, the best approximation
tradeoff for arbitrary metric costs was due to Kuszmaul~\cite{Kuszmaul}.
For equal-length inputs, his randomized algorithm gives an \(O(n^\delta)\)-approximation in
\(\widetilde O(n^{2-\delta})\) time for every fixed \(0<\delta<1\).

\paragraph{Our contribution.}
We present the first constant-factor approximation for weighted edit distance over arbitrary metrics that runs in strongly subquadratic time.

\begin{theorem}\label{thm:main}
For every \(0<\varepsilon\le1\), there is a randomized algorithm that runs
in \(\softO(n^{7/4}/\varepsilon^8)\) time
and computes a \((3+\varepsilon)\)-approximation for \(\ED_\rho(X,Y)\)
with probability at least \(1-n^{-10}\).

\end{theorem}

The algorithm allows unequal input lengths.  The running time is independent
of the alphabet size and the ratio between the largest and smallest positive
metric costs.
\section{Overview}\label{sec:overview}

To efficiently approximate the weighted edit distance, we seek a small collection of local alignments that can be combined into
a global alignment without incurring too much cost.  Following the window
approach of earlier unit-cost approximation algorithms, we
partition \(X\) into short consecutive windows.  A fixed optimal
alignment pairs each of them with a (possibly empty) window of \(Y\).
These \(Y\)-windows occur in order, and the local alignment costs sum
to the optimum \(D\).  We want to find suitable replacements for
these unknown \(Y\)-windows without having to examine them all.

A grid reduces the candidate search by restricting the possible
endpoints.  It must allow the optimal \(Y\)-windows to be rounded
without much additional cost.  Under arbitrary metrics, even moving an
endpoint by a few positions can cross an expensive character.  A bound
on the number of edits therefore gives no bound on this rounding cost.
A long \(Y\)-window can also be close in weighted edit distance to a
short \(X\)-window because it contains many inexpensive characters.
We must control the cost of rounding while keeping these long windows
inexpensive to compare.

\paragraph{Searching by cost.}
The local optimal costs may differ greatly, so we search at increasing
thresholds.  At a threshold \(t\), we seek local alignments of cost at
most \(t\) and keep the rounding and comparison errors small relative
to \(t\).  A box records an \(X\)-window,
a \(Y\)-window, and a label that bounds their weighted edit
distance.  The chunked DP chooses boxes whose windows occur in order in
both strings and pays for skipped characters through insertions and
deletions.  It can combine boxes from different thresholds.  The useful
boxes will have an additional cost proportional to their thresholds, so
we need to find them at thresholds whose sum is close to \(D\).

A guess of the total distance allows us to restrict where candidate windows can
start.  We obtain logarithmically many guesses from a rough estimate
computed using random cost cutoffs and a banded DP.  The estimate is found
by searching the insertion and deletion costs appearing in the input,
which avoids traversing their potentially large numerical range.
For the explanation below, fix a guess between \(D\) and \(2D\).

A character's gap cost is its insertion or deletion cost.  Even a short
window can have arbitrarily large total gap cost while aligning cheaply.
We cap each gap cost at \(t\) and call the resulting sum its clipped
mass.  This bounds an \(X\)-window's clipped mass by its
length times \(t\).  The triangle inequality bounds the difference in
clipped mass between two strings by their weighted edit distance.
For prefixes along the optimal alignment, the global guess bounds this
mass difference and hence restricts candidate starts.  The local
threshold \(t\) similarly restricts ends by bounding the difference
between the \(X\)- and \(Y\)-windows' clipped masses.

We then restrict endpoints to a grid that places cuts on both sides of
expensive characters and groups runs of inexpensive characters into
pieces of small total gap cost.  Rounding inward removes only characters
whose original gap costs sum to \(O(\eps t)\).  The grid also has few
cuts in any interval of clipped prefix mass of width comparable to \(t\).
These properties give small candidate families containing the rounded
optimal \(Y\)-windows once \(t\) covers their local costs.
Grouping \(X\)-windows with similar clipped mass provides a common
candidate family for sampling and sharing comparisons.  Clipped mass is
the key to our weighted framework because it turns bounds on weighted edit
distance into small candidate families whose size is independent of the
range of metric costs.  The edit sequences still pay the original costs.

\paragraph{Comparing long windows.}
A candidate family may have very long windows. To compare them efficiently, 
we adapt Kuszmaul's threshold simplifications~\cite{Kuszmaul} to retain only sufficiently expensive
characters.  We run the weighted DP on the retained subsequences and add
the full gap costs of the removed characters.  The resulting simplified
distance is the cost of an edit sequence and therefore bounds the true
distance from above.

The only loss is the saving from pairing a removed character with a
character of the other string instead of paying both gap costs.
The triangle inequality bounds a pair's saving by twice the smaller gap
cost.
A short \(X\)-window can participate in few pairs, so the total loss is
only \(O(\eps t)\).  A \(Y\)-window within weighted edit distance \(t\) of
the short \(X\)-window also has few retained characters.  Those not matched
to the \(X\)-window must be inserted at their full gap costs.
A count test rejects candidates with too many retained characters before
running the DP.  Every remaining comparison uses a small DP table
regardless of the \(Y\)-window's original length.

\paragraph{Finding useful boxes.}
Even with small candidate families and fast comparisons, examining every
candidate for every \(X\)-window is still too costly.  We develop
weighted dense and sparse phases inspired by unit-cost
approximation algorithms~\cite{CDGKS,AndoniNosatzki,GoldenbergRubinsteinSaha,FMQTW26}.
These phases are designed for the clipped candidate families and the
simplified distance.  
At threshold \(t\), a candidate is accepted if its simplified distance
from the \(X\)-window is at most \(t\).  Sampling from the common
candidate family finds an accepted center for each \(X\)-window with
many accepted candidates.  Windows assigned the same center then share
its comparisons with their candidates.  The union of their candidate
families is also small because their clipped prefix bands overlap at
the center.  The center and a candidate may both be long, but composing
their alignments through the short \(X\)-window gives a matching with
few pairs.  The same bound on lost savings therefore applies to this
comparison.  Once \(t\) covers the local optimal
cost and the small errors, the box for the rounded optimal \(Y\)-window has
label at most the local optimum plus \(2t+O(\eps t)\).  The extra
\(2t\) comes from the alignment of the \(X\)-window with its center.  Its cost is at most
\(t\) and appears once in the label and once in the bound for the
remaining comparison.

The windows left without a center have few accepted candidates, which
a small sample of candidates may miss.  The sparse phase instead samples
\(X\)-windows within consecutive blocks.  The clipped mass of a block
and the cost of its optimal subpath bound how far apart its desired
\(Y\)-window starts can be in clipped prefix mass.  When the subpath is
inexpensive, a neighborhood of one desired start contains the others.
The algorithm cannot identify the desired start among all accepted
starts, so it searches a neighborhood around each of them.  Sparsity
limits the number of neighborhoods, and grid density bounds the number
of candidate starts in each.  We recurse on smaller blocks
until individual \(X\)-windows are reached.
The sparse boxes use direct comparisons and add only \(O(\eps t)\)
simplification cost.

\paragraph{Combining thresholds.}
Consecutive thresholds differ by only a small relative amount.
The first thresholds that cover the local optimal costs and the small
errors therefore sum to nearly \(D\).  However, a desired box may
be missed at that point.
The algorithm keeps boxes from every threshold, and the analysis
chooses a later recovery when necessary.  We must bound the additional
error caused by recovering later.

At threshold \(t\), we charge each missed \(X\)-window an amount \(t\).
The sparse search can miss its desired start for two reasons.  A block's
optimal subpath may be too expensive for localization, or too few of its
windows without a center may have inexpensive local alignments.  We
charge the missed window to the first block where one of these failures occurs.
We bound the charge by the subpath cost in the first case and by the
costs of the expensive local alignments in the second.  None of these
first blocks contains another, so their charges use disjoint portions of
the optimum at a fixed threshold.  Both causes of failure disappear at
the final threshold, so every \(X\)-window is eventually recovered.

The neighborhoods and samples bound the charge at one
threshold by a constant multiple of \(D\) divided by the number of
thresholds.  Summing gives \(O(D)\) total charge for missed windows.  Moving from
\(t\) to the next threshold \((1+\eps)t\) increases each missed window's
threshold by only \(\eps t\).  The total increase in error from the missed windows is
therefore \(O(\eps D)\).  The selected recovery thresholds thus
sum to at most \((1+O(\eps))D\).

We select a recovered box for each nonempty rounded optimal \(Y\)-window.
Each selected \(Y\)-window lies inside its original optimal window, so the boxes
remain ordered even when they come from different thresholds.
When the rounded window is empty, its original \(Y\)-window has small
total gap cost.  We can omit its box and instead delete the \(X\)-window
and insert the original \(Y\)-window.  This costs at most the local optimum
plus a small fraction of the recovery threshold.  After paying for
rounding as well, the chain costs at most \(D\) plus roughly twice the
sum of the recovery thresholds.  This explains where the factor three comes from in
\cref{thm:main}.  The chunked DP finds a chain at least as cheap without
knowing the optimal alignment.

The dense and sparse phases reduce the number of comparisons, and the
simplified distance limits the time spent on each one.  The chunked DP
evaluates only the generated box transitions, so combining thresholds
preserves these running-time savings and gives the bound in
\cref{thm:main}.

\paragraph{Organization.}
\Cref{sec:preliminaries} establishes the notation and basic properties of
weighted edit distance.  \Cref{sec:algorithm} presents the main algorithm
and the chunked DP that combines its boxes.  \Cref{sec:rough} obtains the
rough distance estimate used to limit the number of guesses.
\Cref{sec:simplified} develops the simplified distance for comparing long
windows efficiently, and \cref{sec:grid} constructs the candidate
\(Y\)-windows using clipped mass.  \Cref{sec:threshold} develops the dense
and sparse phases and establishes their sampling and recovery guarantees.
\Cref{sec:proof} combines these results to prove the approximation
guarantee and bound the running time.
\section{Preliminaries}\label{sec:preliminaries}

For a string \(Z\), we write \(Z[j]\) for its character at position \(j\), where
\(0\le j<|Z|\).  A cut \(j\in\{0,\ldots,|Z|\}\) is the boundary after
the first \(j\) characters.  For \(0\le a\le b\le|Z|\), the half-open
fragment \(Z[a\mathinner{.\,.}b)\) contains the characters at positions
\(a,\ldots,b-1\); its length is \(b-a\), and it is empty when \(a=b\).
A \emph{window} \(Z[a\mathinner{.\,.}b)\) consists of this fragment together
with a reference to \(Z\) and its endpoint cuts \(a,b\).  Characters within
a window are indexed locally from zero, so its character at position \(j\)
is \(Z[a+j]\).

We fix a sufficiently large absolute constant \(C\).  Throughout the
analysis, \(\eps\) denotes the accuracy parameter of the algorithm, and we
assume \(0<\eps\le1/C\).  We prove a \((3+O(\eps))\)-approximation, and
\cref{sec:total} rescales \(\eps\) to obtain \cref{thm:main}.  The constant
\(C\) also sets the sampling and localization parameters of
\cref{sec:threshold}.  Every requirement on \(C\) is a lower bound by an
absolute number, so a single constant serves all of them.
The notation \(\softO\) suppresses factors polylogarithmic in \(n/\eps\).

We use a unit-cost
RAM model with oracle access to \(\rho\), in which metric
queries and arithmetic operations or comparisons involving queried costs
take constant time.
For our algorithm, \emph{with high probability} means probability at least
\(1-n^{-10}\).
For every choice of randomness, the output is the cost of an edit sequence
and hence an upper bound on weighted edit distance.

\subsection{Metric edit distance and matchings}

We may restrict the metric to the symbols occurring in \(X\) or \(Y\),
together with \(\gap\). Any intermediate symbol not in the
input can be eliminated by the triangle inequality.  Define the \emph{gap
cost} of symbol \(a\) by
\[
                         g(a)=\rho(a,\gap).
\]
The input strings do not contain \(\gap\). Since \(\rho\) is a metric, the gap
cost of every input character is strictly positive, although there is no
uniform lower bound.
For strings \(U,V\), the alignment graph (edit grid) has
vertices \((i,j)\in\{0,\ldots,|U|\}\times\{0,\ldots,|V|\}\) and forward edges
\begin{align*}
 (i,j)&\to(i+1,j) &&\text{of cost }g(U[i]),
     &&0\le i<|U|,\quad 0\le j\le|V|,\\
 (i,j)&\to(i,j+1) &&\text{of cost }g(V[j]),
     &&0\le i\le|U|,\quad 0\le j<|V|,\\
 (i,j)&\to(i+1,j+1) &&\text{of cost }\rho(U[i],V[j]),
     &&0\le i<|U|,\quad 0\le j<|V|.
\end{align*}
An alignment is a source-to-sink path.  The minimum cost of an alignment is
\(\ED_\rho(U,V)\), abbreviated \(\ED(U,V)\).  The triangle inequality gives
\[
 |g(a)-g(b)|\le\rho(a,b)\le g(a)+g(b).
\]
It also makes \(\ED\) a metric on strings. To see this, compose two alignments by
synchronizing their copies of the intermediate string and apply the triangle
inequality columnwise.

For a string \(Z\), define its \emph{gap mass}
\[
                    M(Z)=\sum_{j=0}^{|Z|-1}g(Z[j]),
\]
and use \(M_Z(j)=M(Z[0\mathinner{.\,.}j))\) for \emph{prefix mass}.  Deleting all of \(U\) and
inserting all of \(V\) costs \(M(U)+M(V)\).

Aligning \(a\) with \(b\) costs \(\rho(a,b)\) instead of \(g(a)+g(b)\).
We call the saving their \emph{matching weight} and define
\[
             w(a,b)=g(a)+g(b)-\rho(a,b).
\]
The triangle inequality implies
\[
                  0\le w(a,b)\le2\min\{g(a),g(b)\}.
\]
A monotone matching is an order-preserving, one-to-one pairing of
characters at specified positions in \(U\) and \(V\).  Paired characters need
not carry the same symbol; a pair represents either a match or a substitution.
The matching's induced edit cost is
\begin{equation}\label{eq:matching-cost}
 M(U)+M(V)-\sum_{(a,b)\in\cM}w(a,b).
\end{equation}
We write \(\operatorname{cost}(\cM)\) for this quantity.
Conversely, every alignment becomes such a matching after deleting its gap
columns.  It follows that
\[
 \ED(U,V)=M(U)+M(V)
   -\max_{\cM\text{ monotone}}\sum_{(a,b)\in\cM}w(a,b).
\]
In particular, a matched pair containing a character \(a\) can save at
most \(2g(a)\).

\subsection{Optimal subpaths and clipped mass}\label{sec:optimal-subpaths}

We assume \(X\ne\varnothing\) for this construction and fix an integer
\(1\le d\le n\), whose value is set in \cref{eq:final-parameters}.
The algorithm partitions \(X\) into
\(q=\lceil |X|/d\rceil\ge1\) consecutive \(X\)-windows \(x_1,\ldots,x_q\), each
of length at most \(d\), with every nonfinal window of length \(d\).  We call any window of \(Y\) a \(Y\)-window.
For the analysis, we let
\(D=\ED(X,Y)\) and fix an optimal path \(\pi\).
At each interior boundary between \(X\)-windows, choose the first vertex
of \(\pi\) with that \(X\)-coordinate.  Together with the source and sink,
these vertices partition \(\pi\) into consecutive subpaths.
These vertices have \(Y\)-coordinates
\[
                     0=z_0^\star\le z_1^\star\le\cdots\le z_q^\star=|Y|.
\]
For each \(i\), let \(y_i=Y[z_{i-1}^\star\mathinner{.\,.}z_i^\star)\) be the optimal \(Y\)-window,
let \(\pi_i\) be the intervening optimal subpath, and put
\(D_i=\operatorname{cost}(\pi_i)\), the cost of this subpath.  Then
\[
       \sum_{i=1}^qD_i=D,\qquad
       \ED(x_i,y_i) = D_i,\qquad
       y_1,\ldots,y_q\text{ occur in order and concatenate to }Y.
\]
The matching induced by each \(\pi_i\) has at most \(|x_i|\le d\) pairs.

An \(X\)-window can have arbitrarily large gap mass even when its length
is at most \(d\).  Capping each gap cost at \(t\) leaves smaller gap costs
unchanged and bounds the clipped mass of an \(X\)-window by \(dt\).
We use clipped mass to control the candidate search and retain the original
costs when evaluating edit sequences.

Fix a threshold \(t>0\) and let \(g_t(a)\) denote the \emph{clipped gap
cost} of a symbol \(a\).  For a string \(Z\) and a cut
\(0\le j\le|Z|\) we define the \emph{clipped mass} \(M_t(Z)\) and
\emph{clipped prefix mass} \(M_{t,Z}(j)\) by
\[
 g_t(a)=\min\{g(a),t\},\qquad
 M_t(Z)=\sum_{j=0}^{|Z|-1}g_t(Z[j]),\qquad
 M_{t,Z}(j)=M_t(Z[0\mathinner{.\,.}j)).
\]
We apply \(M_t\) to \(X\)-windows and \(Y\)-windows and use \(M_{t,X}\) and
\(M_{t,Y}\) for clipped prefix masses in the inputs.

For vertices \(u,v\) appearing in this order on a path \(P\), we write
\(P_{u\to v}\) for the subpath from \(u\) to \(v\), including both endpoints.

Each edit changes the difference between the two clipped prefix masses
by at most its cost.

\Needspace{8\baselineskip}
\begin{lemma}\label{lem:clipped-potential}
For every \(a,b\),
\[
 |g_t(a)-g_t(b)|\le\rho(a,b),\qquad g_t(a)\le g(a).
\]
Consequently, at every vertex \((i,j)\) of any \(X\)-to-\(Y\) edit path
\(P\),
\[
 \left|M_{t,Y}(j)-M_{t,X}(i)\right|
 \le \operatorname{cost}(P_{(0,0)\to(i,j)}).
\]
For a subpath from \((i,j)\) to \((i',j')\),
\[
 \left|M_t(X[i\mathinner{.\,.}i'))-M_t(Y[j\mathinner{.\,.}j'))\right|
 \le\operatorname{cost}(P_{(i,j)\to(i',j')}).
\]
\end{lemma}
\begin{proof}
Truncation \(z\mapsto\min\{z,t\}\) is one-Lipschitz on
\(\R_{\ge0}\), and the triangle inequality gives
\(|g(a)-g(b)|\le\rho(a,b)\).  Along a deletion, insertion, or substitution edge, the
absolute change of \(M_{t,Y}-M_{t,X}\) is at most that edge's cost.
Telescoping proves both path statements.
\end{proof}

We use the prefix form of the lemma to restrict candidate starts in \(Y\) using a
global guess \(\Delta\), and its subpath form to constrain the clipped mass of a
\(Y\)-window using the cost \(D_i\) of the optimal subpath.
\section{Main algorithm}\label{sec:algorithm}

We present the full algorithm (\cref{alg:master}) using boxes to represent local alignments.
A box records an \(X\)-window, a \(Y\)-window, and an upper bound on the
cost of aligning them.  A \emph{box} for \(x_i\) is a tuple
\[
             r=(i,a_r,e_r,\lambda_r),
             \qquad 0\le a_r\le e_r\le |Y|,
\]
whose label satisfies
\[
                   \lambda_r\ge\ED(x_i,Y[a_r\mathinner{.\,.}e_r)).
\]
The chunked DP selects boxes whose windows occur in order in both
strings.  It pays their labels, deletes skipped \(X\)-windows, and inserts
uncovered characters of \(Y\).

\subsection{Guesses, thresholds, and the main loop}\label{sec:main-loop}

We use the partition \(x_1,\ldots,x_q\) of \(X\) into windows of length
at most \(d\) from \cref{sec:optimal-subpaths}.  A guess \(\Delta\) of the
total distance \(D\) restricts where their candidate \(Y\)-windows can
start.  We obtain logarithmically many guesses by repeatedly halving a
positive estimate \(\widehat D\) from \emph{RoughUpperBound}, forming the
family \(\mathcal G(\widehat D)\) specified in \cref{sec:rough}.  For each
guess, we search at the geometrically increasing thresholds
\(t_0,\ldots,t_L\) from \cref{eq:threshold-ladder}.  The accuracy
parameter \(\eps\) controls their spacing.  Increasing
the threshold allows the recovery of \(X\)-windows whose optimal subpaths
are too expensive at earlier thresholds.

At a fixed threshold \(t\), we build the grid \(G_t\) and the global tag
family \(\cY_t^{\rm tag}\) as in \cref{sec:grid}.  The \emph{prefix data}
consist of the clipped prefix masses of \(X\) and \(Y\), the \(X\)-window
bins, and the prefix counts and positions of retained characters in \(Y\)
from \cref{sec:simplified}.  The acceptance test from
\cref{sec:simplified,sec:grid} guarantees an edit sequence of cost at most
\(t\) from an \(X\)-window to each of its accepted candidates.
\emph{DensePhase} (\cref{alg:dense-phase}) uses sampled candidates as
centers, so that windows assigned the same center can share its
comparisons with their candidates.  These shared comparisons give the
dense boxes, which the phase returns as the list \(\mathcal D_t\) together
with the index set \(\cX'\) of windows left without a center.

\emph{SparsePhase} (\cref{alg:sparse-phase}) searches for boxes for these
residual windows while keeping \(\cX'\) fixed.  The recursion starts with
all \(X\)-window indices and \(G_t\) as its list of candidate starts.
It divides the index interval into consecutive blocks and processes only
the residual windows within each block.  It returns the list
\(\mathcal S_t\) of boxes found by the recursion.  This search
relies on each residual window being \(t\)-sparse, that is, having fewer
than \(k\) accepted tags (\cref{sec:threshold}).
To bound its running time on every execution, a call to
\emph{ProcessXWindow} abandons the guess if it finds its \(k\)-th distinct
accepted tag.  The algorithm then discards that guess's boxes and
continues with the next guess, retaining the best value already
obtained.  Both phases use the sampling and localization parameters of
\cref{sec:threshold}; the values of \(d,k\) are chosen in
\cref{sec:total}.

We collect in \(\mathcal B\) the boxes in \(\mathcal D_t\) and
\(\mathcal S_t\) from successive thresholds.  Once the threshold sequence for a guess is complete,
\emph{ChunkedDP}\((\mathcal B)\) returns the minimum chain cost among these
boxes.  The chain can therefore combine boxes found
at different thresholds.  We keep the smallest value over completed
guesses, starting with the cost of the all-gap edit sequence that deletes
every character of \(X\) and inserts every character of \(Y\).  This is
also the value returned when either input is empty.

\begin{algorithm}[H]
\caption{Weighted edit distance approximation}\label{alg:master}
\small
\begin{algorithmic}[1]
\Require strings \(X,Y\), metric \(\rho\), and accuracy parameter \(0<\eps\le1/C\)
\State \(V\gets M(X)+M(Y)\)
\If{\(X\) or \(Y\) is empty}
  \State \Return \(V\)
\EndIf
\State Set \(d,k\) and the sampling and localization parameters
\State Partition \(X\) into windows of length \(d\)
\State \(\widehat D\gets\Call{RoughUpperBound}{X,Y}\)
\If{\(\widehat D=0\)}
  \State \Return \(0\)
\EndIf
\Statex
\For{\(\Delta\in\mathcal G(\widehat D)\)}
  \State Set the thresholds \(t_0,\ldots,t_L\) for \(\Delta\)
  \State \(\mathcal B\gets\varnothing\)
  \For{\(t=t_0,t_1,\ldots,t_L\)}
    \State Build \(G_t\), \(\cY_t^{\rm tag}\), and the prefix data for \(t\)
    \State \((\mathcal D_t,\cX')\gets\Call{DensePhase}{\Delta,t}\)
    \State \(\mathcal S_t\gets\Call{SparsePhase}{\Delta,t,\cX'}\)
    \State Add the boxes in \(\mathcal D_t\) and \(\mathcal S_t\) to \(\mathcal B\)
  \EndFor
  \Statex
  \State \(V_\Delta\gets\Call{ChunkedDP}{\mathcal B}\)
  \State \(V\gets\min\{V,V_\Delta\}\)
\EndFor
\Statex
\State \Return \(V\)
\end{algorithmic}
\end{algorithm}

\subsection{Boxes and the chunked DP}\label{sec:aggregation}

The chunked DP coarsens the edit-distance DP to the \(X\)-windows
and uses boxes as shortcut transitions.
A threshold grid may contain every cut of \(Y\).  The full
coarsened table can therefore have \(\Theta(n^2/d)\) cells, which exceeds
the running-time budget for our choice of \(d\).  We use sparse dynamic
programming to evaluate only the box transitions
(\cref{lem:chain-dp}).

A nonempty \emph{box chain} \(\mathcal C=(r_1,\ldots,r_m)\), where
\(r_j=(i_j,a_j,e_j,\lambda_j)\), satisfies
\[
 i_1<\cdots<i_m,
 \qquad e_j\le a_{j+1}\quad(1\le j<m).
\]
The chain pays its box labels, deletes skipped \(X\)-windows, and inserts
the portions of \(Y\) before, between, and after its chosen windows.
To express the deletion costs, put
\[
 \Gamma_i=\sum_{i'\le i}M(x_{i'}),\qquad \Gamma_0=0.
\]
Thus \(\Gamma_i\) is the cost of deleting the first \(i\) \(X\)-windows,
and the chain cost is
\begin{align}
 \operatorname{cost}(\mathcal C)
 ={}&\Gamma_{i_1-1}+M_Y(a_1)+\lambda_1 \notag\\
 &+\sum_{j=2}^{m}\left[
       \Gamma_{i_j-1}-\Gamma_{i_{j-1}}
       +M_Y(a_j)-M_Y(e_{j-1})+\lambda_j\right] \notag\\
 &+\Gamma_q-\Gamma_{i_m}+M_Y(|Y|)-M_Y(e_m).
 \label{eq:chain-cost}
\end{align}
The empty chain has cost
\(\Gamma_q+M_Y(|Y|)=M(X)+M(Y)\).

\begin{lemma}\label{lem:chain-sound}
Every box chain defines an edit sequence from \(X\) to \(Y\) of cost
at most its chain cost.  Consequently,
\[
       \ED(X,Y)\le\min_{\mathcal C}\operatorname{cost}(\mathcal C)
\]
for every box collection.
\end{lemma}
\begin{proof}
For each box, use a local edit sequence of cost at most its label.
Delete every \(X\)-window not represented in
the chain, and insert the substrings of \(Y\) before, between, and after the
chosen \(Y\)-windows.  Strictly increasing \(X\)-window indices ensure that
the local edit sequences act on disjoint \(X\)-windows.  Since the chosen \(Y\)-windows
are disjoint and ordered, these insertions cover each remaining
character of \(Y\) exactly once.  The resulting
edit sequence has cost at most the chain cost in \cref{eq:chain-cost}.
\end{proof}

\begin{lemma}\label{lem:chain-dp}
Given \(P_{\rm box}\) boxes enumerated by increasing \(X\)-window index,
the minimum box-chain cost can be computed in \(O(n+P_{\rm box}\log n)\) time.
\end{lemma}
\begin{proof}
Consider appending a box \(r=(i,a_r,e_r,\lambda_r)\) to a chain whose last
box belongs to an earlier \(X\)-window \(x_{i'}\), with \(i'<i\), and ends at
a cut \(e\le a_r\).  Let \(\Phi\) be the chain's prefix cost, excluding the
final deletions and insertions.  Appending \(r\) deletes the skipped
\(X\)-windows, inserts \(Y[e\mathinner{.\,.}a_r)\), and applies the local
edit sequence of \(r\), so the extended chain has prefix cost
\[
 \bigl(\Gamma_{i-1}+M_Y(a_r)+\lambda_r\bigr)
      +\bigl(\Phi-\Gamma_{i'}-M_Y(e)\bigr).
\]
The first term depends only on \(r\).  After subtracting the prefix gap
masses, the predecessor contributes the same second term to every box
that can follow it.  We therefore need only the minimum of that term
over predecessors ending at a cut \(e\le a_r\).
For each cut \(e\), let \(K(e)\) be the minimum of \(\Phi-\Gamma_{i'}-M_Y(e)\)
over the chains found so far whose last box ends at \(e\); initially
\(K(0)=0\), which represents the empty chain, and \(K(e)=+\infty\) for
\(e>0\).  We process the \(X\)-windows in order.  For each box \(r\) of \(x_i\),
the cheapest prefix cost of a chain ending at \(r\) is
\begin{equation}\label{eq:box-recurrence}
 \Phi(r)=\Gamma_{i-1}+M_Y(a_r)+\lambda_r+\min_{0\le e\le a_r}K(e).
\end{equation}
After all boxes of \(x_i\) are evaluated, we set
\(K(e_r)\gets\min\{K(e_r),\Phi(r)-\Gamma_i-M_Y(e_r)\}\) for each of them.
Deferring these updates enforces \(i'<i\).  After the last \(X\)-window, the
minimum chain cost is \(\Gamma_q+M_Y(|Y|)+\min_eK(e)\).  A Fenwick tree over
the cuts of \(Y\) supports these point-minimum updates and prefix-minimum
queries in \(O(\log n)\) time each, which gives the stated bound.
\end{proof}

Given \(\mathcal B\), \emph{ChunkedDP}\((\mathcal B)\) groups its boxes by
\(X\)-window index in \(O(q+|\mathcal B|)\) time and evaluates them with the
DP of \cref{lem:chain-dp}, which returns the minimum box-chain cost,
including the empty chain.  \(\mathcal B\) may contain several boxes for the
same pair of windows, found in different calls or at different thresholds.
They need no special treatment because the DP takes minima.
\section{A rough distance estimate}\label{sec:rough}

The main algorithm needs a rough upper bound to locate useful distance
guesses without searching through the numerical range of the edit costs.
We obtain it by removing characters at or below a random gap-cost cutoff,
running banded dynamic programming on the retained strings, and adding
the full gap costs of the removed characters.

If either input is empty, the routine returns \(M(X)+M(Y)\).
For the remaining case, we fix a failure probability \(0<\delta_{\rm rg}<1\).

\begin{theorem}\label{thm:rough-guide}
The routine below always returns a value \(\widehat D\ge D\) that is the cost of an
edit sequence.  If \(D=0\), it returns zero deterministically.  If
\(D>0\), then with probability at
least \(1-\delta_{\rm rg}\),
\[
 \widehat D<32(1+n)D.
\]
Its running time is
\[
 O\!\left(n\log n
       \log\frac{n+1}{\delta_{\rm rg}}\right).
\]
\end{theorem}

For the main algorithm, we take \(\delta_{\rm rg}=1/(2n^{10})\) and call
this specialization \emph{RoughUpperBound}.  The routine takes
\(\softO(n)\) time.  If \(D>0\), its estimate satisfies
\[
                         \widehat D<32(1+n)D\le64nD
\]
with probability at least \(1-\delta_{\rm rg}\).
For a positive estimate \(\widehat D\), the main algorithm considers
the guesses
\[
 \mathcal G(\widehat D)=
 \left\{\frac{\widehat D}{2^\nu}:0\le \nu\le\lceil\log_2(64n)\rceil+1\right\}.
\]
This family has \(O(\log n)\) members.  Whenever \(\widehat D<64nD\),
it contains a \emph{tight guess} \(D\le\Delta<2D\).

\subsection{Threshold simplification and exact bands}

For \(\theta\ge0\), let \(X^{>\theta},Y^{>\theta}\) be Kuszmaul's
threshold simplifications~\cite{Kuszmaul} of \(X,Y\), obtained by removing all
characters with gap cost at most \(\theta\) while preserving the remaining order.
Let \(M_{\rm cut}(\theta)\) be the total gap cost of the characters removed
from \(X\) and \(Y\) and put \(D(\theta)=\ED(X^{>\theta},Y^{>\theta})\).
Then \(M_{\rm cut}(\theta)\le n\theta\).  To extend an edit sequence on the
retained strings to the original strings, delete the removed characters of
\(X\) and insert the removed characters of \(Y\) in their original order.
This adds exactly \(M_{\rm cut}(\theta)\) to the cost.

The indices in the edit grid for \(X^{>\theta},Y^{>\theta}\) refer to positions in the
retained strings.  We restrict this grid to the band \(|i-j|\le1\), and let
\(D_{\rm band}(\theta)\) be the cost of its shortest path, or \(+\infty\)
if its sink is unreachable.  The standard weighted recurrence evaluates
this band in \(O(n)\) time.

\begin{lemma}\label{lem:rg-band}
For every \(\theta>0\),
\(D_{\rm band}(\theta)\ge D(\theta)\).  If
\(D(\theta)\le\theta\), equality holds.  More generally, if all characters of
\(X^{>\theta},Y^{>\theta}\) have gap cost at least \(R>0\), equality holds whenever
\(D(\theta)\le R\).
\end{lemma}
\begin{proof}
Restricting the grid cannot lower the optimum.  Every gap edge in the edit
grid for the retained strings costs more than \(\theta\).  A path of cost at most
\(\theta\) therefore has no gap edges.  Since \(i-j\) changes only on a gap
edge, the whole path lies in the band.  If the retained characters instead
have gap cost at least \(R\), a path of cost at most \(R\) has at most one
gap edge and is still contained in the inclusive band.
\end{proof}

\subsection{Random cutoffs}

Removing characters can separate a matched pair and replace its
substitution cost by a gap cost.  We control the expected increase in
cost by sampling the cutoff from a finite grid in \([R,2R)\) at a scale
\(R>0\).

Let \(G\) be the least power of two satisfying \(G\ge2n\).
We draw \(J\) uniformly from
\(\{0,\ldots,G-1\}\) and put
\[
                         \theta_J=R(1+J/G).
\]
An attempt at scale \(R\) draws a fresh cutoff \(\theta_J\) and computes
\(D_{\rm band}(\theta_J)\).
It succeeds when \(D_{\rm band}(\theta_J)\le\theta_J\) and returns
\(D_{\rm band}(\theta_J)+M_{\rm cut}(\theta_J)\).  The banded DP is exact whenever
the retained strings have distance at most this acceptance threshold.

We bound the expected distance between the retained strings to control
the probability that an attempt fails.
The choice of \(G\) makes the discretization term \(nR/G\) in the
next lemma at most \(R/2\).  Taking a power of two allows exact
uniform sampling with \(\log_2G\) random bits.

\begin{lemma}\label{lem:rg-projection}
\[
                         \E_J[D(\theta_J)]
                         \le3D+\frac{nR}{G}.
\]
\end{lemma}
\begin{proof}
Fix an optimal alignment and project it to the retained strings.  A gap column
never becomes more expensive.  Consider a nongap column \((a,b)\), and
assume without loss of generality that \(g(a)\le g(b)\).  Its endpoints are
separated only when \(g(a)\le\theta_J<g(b)\).  Since the cutoff spacing is
\(R/G\) and \(g(b)-g(a)\le\rho(a,b)\),
\[
 \Prb[g(a)\le\theta_J<g(b)]\le\rho(a,b)/R+1/G.
\]
If separated, \(b\) is aligned with a gap.  Relative to the original cost
\(\rho(a,b)\), the additional cost is \(g(b)-\rho(a,b)\le g(a)<2R\).  Thus
the expected projected cost of this column is at most \(3\rho(a,b)+2R/G\).
There are at most \(n/2\) nongap columns.  Summing proves the claim.
\end{proof}

\begin{lemma}\label{lem:rg-one-scale}
Every successful attempt at scale \(R\) returns the cost of an edit
sequence, and this value is smaller than \(2(1+n)R\).
If \(R\ge8D\), an attempt succeeds with probability at least
\(1/8\).
\end{lemma}
\begin{proof}
When an attempt succeeds, \cref{lem:rg-band} gives \(D_{\rm band}(\theta_J)=D(\theta_J)\).
Lifting the edit sequence proves feasibility.
The inequality \(\theta_J<2R\) gives the stated bound on the returned value.
A failed attempt has
\(D(\theta_J)>\theta_J\ge R\).  When \(R\ge8D\),
\cref{lem:rg-projection} and the choice of \(G\) give
\(\E[D(\theta_J)]\le7R/8\).
Markov's inequality therefore bounds the failure probability by \(7/8\),
proving the claim.
\end{proof}

We repeat attempts to make it unlikely that any sufficiently large tested
scale fails.  The search below tests at most \(n+1\) scales, so it
suffices to bound the failure probability of each such test by
\(\delta_{\rm rg}/(n+1)\).  Let \(m_{\rm rg}\) be the least positive
integer satisfying
\[
 \left(\frac78\right)^{m_{\rm rg}}
 \le\frac{\delta_{\rm rg}}{n+1}.
\]
A call \emph{TestScale}\((R)\) makes up to \(m_{\rm rg}\) independent
attempts and returns the value from the first successful attempt.
If every attempt fails, it returns \textsc{Fail}.  Thus, whenever
\(R\ge8D\), the call returns a value with probability at least
\(1-\delta_{\rm rg}/(n+1)\).

\subsection{Searching the input gap costs}

We cache and sort the distinct gap costs appearing in \(X\) and \(Y\),
\(g_{(1)}<\cdots<g_{(N)}\), and put \(g_{(0)}=g_{(1)}/2\).  Bisection acts on indices in
the list \(g_{(0)},\ldots,g_{(N)}\).  A test at index \(j\) means a call to
\emph{TestScale}\((g_{(j)})\).  Searching indices uses only
\(O(\log n)\) tests, regardless of the ratio between the gap costs.
However, adjacent list values may still differ by an arbitrarily large
factor.  The final step handles this case by using the fact that no
retained string changes between consecutive input gap costs.

The routine proceeds as follows.

\begin{enumerate}[leftmargin=*]
\item If \(X=Y\), return zero.  Otherwise \(D>0\) because \(\rho\)
is a metric.
\item Test \(g_{(0)}\).  If it returns a value, return it; otherwise mark index
zero \textsc{Fail}.  Test \(g_{(N)}\), which removes every character and
therefore returns a value.
\item Maintain a lower index labeled \textsc{Fail} and an upper index with
a stored successful value.  Repeatedly bisect this index interval.  If a test
returns \textsc{Fail}, replace the lower index; otherwise replace the upper
index and store the returned value.  No monotonicity of the random outcomes
is assumed.
Stop at adjacent indices \(j,j+1\) and set \(R_L=g_{(j)}\) and \(R_U=g_{(j+1)}\).
\item If \(R_U\le2R_L\), return the stored value from scale \(R_U\).
Otherwise compute \(D_{\rm band}(R_L)\) on the strings simplified at \(R_L\).
If \(D_{\rm band}(R_L)\le R_U\), return
\(D_{\rm band}(R_L)+M_{\rm cut}(R_L)\).  Otherwise return the stored value
from scale \(R_U\).
\end{enumerate}

\begin{proof}[Proof of \cref{thm:rough-guide}]
The lifting construction and \cref{lem:rg-band} show that all returned
values are costs of edit sequences.  The explicit equality test returns
zero whenever \(D=0\).  Suppose now that \(D>0\).

Conditional on all preceding tests, the next tested scale is fixed and
\emph{TestScale} uses fresh randomness.  By
\cref{lem:rg-one-scale} and the choice of \(m_{\rm rg}\), its failure
probability is at most \(\delta_{\rm rg}/(n+1)\) whenever
\(R\ge8D\).  Each of the at most \(n+1\) list entries is
tested at most once.  A union bound therefore shows that all tested
scales \(R\ge8D\) return values with probability at least
\(1-\delta_{\rm rg}\).  Assume this event for the remainder of the proof.
Every tested scale whose call returns \textsc{Fail} then satisfies
\[
                         R<8D.
\]
This is the only property of failed tests that the search needs;
successful and failed outcomes need not be monotone in \(R\).

Every cutoff used by \emph{TestScale}\((g_{(0)})\) is smaller than
\(g_{(1)}\), so no character is removed and a returned value is exact.
Otherwise, bisection ends with a failed lower scale \(R_L\) and a
successful upper scale \(R_U\).  If \(R_U\le2R_L\), then \(R_U<16D\),
so \cref{lem:rg-one-scale} gives \(\widehat D<32(1+n)D\).

Suppose \(R_U>2R_L\).  Since \(g_{(1)}=2g_{(0)}\), this case has \(R_L\ne g_{(0)}\).
Thus \(R_L,R_U\) are consecutive input gap costs.  All cutoffs in
\([R_L,R_U)\) produce the same retained strings, and every character in
those strings has gap cost at least \(R_U\).  Every cutoff sampled at scale
\(R_L\) lies in \([R_L,2R_L)\subset[R_L,R_U)\).  The retained strings are
therefore the same for every sampled cutoff, and the expectation in
\cref{lem:rg-projection} is exactly \(D(R_L)\).  That lemma and the bound
on the failed scale give
\[
 D(R_L)<7D,
 \qquad
 M_{\rm cut}(R_L)<8nD.
\]

The final band computation uses these retained strings.  If
\(D_{\rm band}(R_L)\le R_U\), \cref{lem:rg-band} with minimum retained gap cost
\(R_U\) gives \(D_{\rm band}(R_L)=D(R_L)\).  The returned value is therefore smaller than
\((7+8n)D\).
If \(D_{\rm band}(R_L)>R_U\), the same lemma implies \(D(R_L)>R_U\); otherwise
the banded DP would return \(D(R_L)\).  Hence \(R_U<D(R_L)<7D\), so the
stored successful value is smaller than \(14(1+n)D\) by
\cref{lem:rg-one-scale}.  Both outcomes satisfy the theorem's bound.

Sorting costs \(O(n\log n)\).  Bisection tests \(O(\log n)\) indices,
each with \(O(\log((n+1)/\delta_{\rm rg}))\) banded computations.
Each computation, including simplification and adding the removed
characters' gap costs, takes \(O(n)\) time.  This proves the running-time
bound.
\end{proof}
\section{Simplified distance}\label{sec:simplified}

A \(Y\)-window close in weighted edit distance to a short \(X\)-window
may still contain many inexpensive characters.  The standard DP would
have to process all of them.  We discard the matching savings of
sufficiently inexpensive characters while retaining their full gap costs.
For an \(X\)-window of length at most \(d\), this increases the minimum
alignment cost by only \(O(\eps t)\).  A count test rejects windows with
too many retained characters, so each comparison with an \(X\)-window
takes \(O(d^2/\eps)\) time.

Throughout, we fix a distance threshold \(t>0\) and an integer
\(d\ge1\).  We set
\(H=\eps t/d\).  For a string \(U\), let \(U^{>H}\) be the indexed
subsequence, in original order, of characters whose gap cost is strictly
greater than \(H\); we call these the retained characters.

\begin{definition}\label{def:simplified}
For strings \(U,V\), the \emph{simplified distance} at threshold \(t\) is
\[
 \SD_t(U,V)=M(U)+M(V)
 -\max_{\cM\subseteq U^{>H}\times V^{>H}}
                  \sum_{(a,b)\in\cM}w(a,b),
\]
where the maximum is over monotone one-to-one matchings.
\end{definition}
Since \(\ED\) is the total gap mass minus the best matching savings,
\[
 \SD_t(U,V)=\ED\bigl(U^{>H},V^{>H}\bigr)+\bigl(M(U)-M(U^{>H})\bigr)
            +\bigl(M(V)-M(V^{>H})\bigr).
\]
We compute the simplified distance by running the standard edit-distance
DP on the retained strings, then adding the costs of deleting the removed
characters of \(U\) and inserting the removed characters of \(V\).  The subsequences
\(U^{>H},V^{>H}\) are Kuszmaul's threshold simplifications~\cite{Kuszmaul}.
Prefix gap-mass arrays provide \(M(U)\) and \(M(V)\).  Subtracting the
retained characters' gap costs gives the removed cost without scanning
the entire substrings.

If \(|U|\le d\), an optimal matching has at most \(d\) pairs.  Discarding
pairs that involve a character of gap cost at most \(H\) loses at most
\(2H\) per pair.

\begin{lemma}\label{lem:simplified-short}
Let \(\cM\) be any monotone matching between \(U,V\) with at most \(d\)
pairs.  Then
\[
 \ED(U,V)\le\SD_t(U,V)
 \le\operatorname{cost}(\cM)+2\eps t.
\]
In particular, if \(|U|\le d\),
\[
 \ED(U,V)\le\SD_t(U,V)\le\ED(U,V)+2\eps t.
\]
The maximizing restricted matching induces an edit sequence whose cost is
the simplified distance.
\end{lemma}
\begin{proof}
A maximizing matching induces an edit sequence whose cost is the simplified
distance, proving the lower bound.  From \(\cM\),
remove every pair with at least one endpoint of gap cost at most \(H\).
Since \(w(a,b)\le2\min\{g(a),g(b)\}\), each removed pair has matching weight
at most \(2H\).  There are at most \(d\) such pairs, so
\cref{eq:matching-cost} shows that removing them increases the induced edit
cost by at most \(2dH=2\eps t\).  The remaining matching is allowed in
\cref{def:simplified}, proving the upper bound.  For the bound with
\(|U|\le d\), use an optimal matching, which has at most \(|U|\) pairs.
\end{proof}

Thus, whenever \(|U|\le d\),
\[
 \ED(U,V)\le(1-2\eps)t\ \Longrightarrow\ \SD_t(U,V)\le t,
 \qquad
 \SD_t(U,V)\le t\ \Longrightarrow\ \ED(U,V)\le t.
\]

\subsection{Bounding retained characters}

At each threshold the algorithm stores prefix counts and ordered positions for
the retained characters of \(Y\), those with gap cost exceeding \(H\).
Before every comparison with a \(Y\)-window \(V\), it applies the
\emph{count test}, checking in constant time whether
\(\lvert V^{>H}\rvert\le d+\lceil d/\eps\rceil\).  The following lemma shows
that every \(Y\)-window within weighted edit distance \(t\) of an
\(X\)-window passes this test.  The test counts retained characters in \(V\), whose number
is distinct from the clipped mass used in \cref{sec:grid}.

\begin{lemma}\label{lem:retained-count}
Suppose a monotone matching between \(U,V\) has at most \(d\) pairs and
induced edit cost at most \(\kappa t\), for some \(\kappa\ge0\).  Then
\[
                  |V^{>H}|\le d+\frac{\kappa t}{H}
                            =d+\frac{\kappa d}{\eps}.
\]
In particular, the conclusion holds if \(|U|\le d\) and
\(\ED(U,V)\le\kappa t\).
\end{lemma}
\begin{proof}
At most \(d\) characters of \(V\) are matched.  Every remaining character
of \(V^{>H}\) is inserted by the induced edit sequence and costs more than \(H\).
Their number is at most \(\kappa t/H\).  For the final statement, use an
optimal matching, which has at most \(|U|\le d\) pairs.
\end{proof}

When comparing an \(X\)-window with a \(Y\)-window that passes the count test,
the simplified-distance DP table has dimensions \(O(d)\) and \(O(d/\eps)\)
and can be computed in \(O(d^2/\eps)\) time.  The \(Y\)-window may contain
arbitrarily many characters with gap cost at most \(H\).  Their gap costs
are included in the simplified distance, but these characters do not
enter the DP table.

\subsection{Composing comparisons through an \texorpdfstring{\(X\)}{X}-window}

The center and the candidate \(Y\)-window in the dense phase may both be
long, but each is matched to the same short \(X\)-window.  Composing these
matchings produces a center-to-candidate matching with at most \(d\) pairs.

\begin{lemma}\label{lem:common-source}
Let \(\cM_{U,C}\) and \(\cM_{U,V}\) be monotone matchings from the same
string \(U\) into \(C\) and \(V\).  Pair \(c\in C\) with \(v\in V\) exactly
when they are matched to the same position in \(U\).  The resulting
matching \(\cM_{C,V}\) is monotone, has at most \(|U|\) pairs, and satisfies
\[
 \operatorname{cost}(\cM_{C,V})
 \le\operatorname{cost}(\cM_{U,C})+
     \operatorname{cost}(\cM_{U,V}).
\]
\end{lemma}
\begin{proof}
Monotonicity follows because both original matchings preserve the order of
positions in \(U\).  Moreover, every pair in \(\cM_{C,V}\) is indexed
by a distinct position in \(U\), so \(\cM_{C,V}\) has at most \(|U|\)
pairs.

By \cref{eq:matching-cost}, the desired inequality is equivalent to
\[
 \sum_{(u,c)\in\cM_{U,C}} w(u,c)
 +\sum_{(u,v)\in\cM_{U,V}} w(u,v)
 -\sum_{(c,v)\in\cM_{C,V}} w(c,v)
 \le 2M(U).
\]
Charge the left-hand side to positions in \(U\).  Consider a position with
character \(u\).  If it is unmatched in both original matchings, its
contribution is zero.  If it is matched in exactly one of them, its
contribution is the matching weight of that pair and is at most \(2g(u)\).
If \(u\) is matched to \(c\) in the first
matching and to \(v\) in the second, its contribution is
\[
 \begin{aligned}
 w(u,c)+w(u,v)-w(c,v)
   &=2g(u)-\rho(u,c)-\rho(u,v)+\rho(c,v)\\
   &\le 2g(u),
 \end{aligned}
\]
where the inequality follows from
\(\rho(c,v)\le\rho(c,u)+\rho(u,v)\) and symmetry of \(\rho\).
Summing over all positions in \(U\) gives a bound of \(2M(U)\), proving
the claim.
\end{proof}

\Needspace{8\baselineskip}
\begin{corollary}\label{cor:dense-composition}
Let \(|U|\le d\) and \(D_V=\ED(U,V)\).  Assume that
\(\SD_t(U,C)\le t\) and \(D_V\le t\).  Then
\[
 \SD_t(C,V)\le \SD_t(U,C)+D_V+2\eps t.
\]
The combined box label is the cost of a \(U\)-to-\(V\) edit sequence
and satisfies
\[
               \SD_t(U,C)+\SD_t(C,V)\le D_V+(2+2\eps)t.
\]
\end{corollary}
\begin{proof}
\setlength{\emergencystretch}{1em}
Use a matching realizing \(\SD_t(U,C)\) and an optimal
\(U\)-to-\(V\) matching in \cref{lem:common-source}.  Their composition
has at most \(d\) pairs and cost at most \(\SD_t(U,C)+D_V\).  Apply
\cref{lem:simplified-short}; then use \(\SD_t(U,C)\le t\) once more.  The edit sequences
realizing \(\SD_t(U,C)\) and \(\SD_t(C,V)\)
concatenate, proving the feasibility statement.
\end{proof}

The dense box label includes \(\SD_t(U,C)\) for the comparison with the center, and
the bound on the center-to-candidate comparison contains another copy of
\(\SD_t(U,C)\). This explains the overhead of about \(2t\) in the
label bound of \cref{cor:dense-composition}.

\paragraph{Size of the center-to-candidate computation.}
Before a center-to-candidate comparison, the algorithm applies the same count
test to the candidate \(V\); the center \(C\) passed it when it was
selected.  Both retained strings have length \(O(d/\eps)\), so every
comparison takes \(O(d^2/\eps^2)\) time.
\section{Candidate \texorpdfstring{\(Y\)}{Y}-windows}\label{sec:grid}

Candidate \(Y\)-windows for an \(X\)-window must have starts in a band of
clipped prefix mass and clipped masses close to that of the \(X\)-window.
The guess \(\Delta>0\) controls the start band and the threshold \(t>0\)
controls the allowed difference in clipped mass.  We set the mesh
\(s=\eps t/64\) so that rounding removes only \(O(\eps t)\) gap mass.

\subsection{Clipped grids}

A naive clipped-mass grid need not include both cuts adjacent to a
character of \(Y\) whose clipped gap cost exceeds the mesh, so an endpoint
snap could remove characters of arbitrarily large unclipped gap cost.
Our grid isolates every such character before greedily grouping the remaining ones.

We call a character \(Y[j]\) \emph{\(s\)-heavy} when \(g_t(Y[j])>s\), and
\emph{\(s\)-light} otherwise.  We mark both adjacent cuts of every \(s\)-heavy
character, together with \(0\) and \(|Y|\).  Within each maximal remaining run, all clipped
gap costs equal their corresponding gap costs and are at most \(s\)
because \(g_t(Y[j])\le s<t\) implies \(g_t(Y[j])=g(Y[j])\).  We scan each such run from left
to right, marking the first cut at which the accumulated mass reaches
\(s\), restarting there, and marking the run's final cut.  Let \(G_t\) be
the sorted set of marked cuts.  The construction takes one scan through
\(Y\).

Any two consecutive grid cells have total clipped mass at least \(s\).
This bounds the number of cuts in an interval of clipped prefix values.

\begin{lemma}\label{lem:grid-density}
For every interval \(I\subseteq\R\), writing \(|I|\) for its length,
\[
 \bigl|\{z\in G_t:M_{t,Y}(z)\in I\}\bigr|
 \le3+\frac{2|I|}{s}=O\!\left(1+\frac{|I|}{\eps t}\right).
\]
In particular \(|G_t|=O(n)\).
\end{lemma}
\begin{proof}
A grid cell is the substring between two consecutive cuts of \(G_t\).
Every greedy cell whose construction reaches mass \(s\) has clipped mass
in \([s,2s)\).  Immediately before its last character is added, the
accumulated mass is below \(s\), and that character contributes at most
\(s\).  A cell consisting of an isolated \(s\)-heavy character has mass
greater than \(s\).  Thus a cell of mass below \(s\) can occur only as the
final, unfinished cell of a run of \(s\)-light characters.  Each run has
at most one such cell, and an isolated \(s\)-heavy character separates
distinct runs.  In particular, two cells of mass below \(s\) cannot be
adjacent.  Every pair of consecutive grid cells therefore has total clipped
mass at least \(s\).

Let \(z_1<\cdots<z_m\) be all cuts in \(G_t\) whose clipped prefix values
lie in \(I\).  The claim is immediate if \(m\le1\), so assume \(m\ge2\).
Clipped prefix mass is strictly increasing, and \(I\) is an interval, so
these are consecutive cuts of \(G_t\).  The \(m-1\) cells between them
have total clipped mass
\[
                 M_{t,Y}(z_m)-M_{t,Y}(z_1)\le |I|.
\]
Pairing consecutive cells gives
\[
                 \left\lfloor\frac{m-1}{2}\right\rfloor s\le |I|.
\]
Therefore \(m\le3+2|I|/s\).

Finally, \(G_t\) is a set of cuts of \(Y\), so \(|G_t|\le |Y|+1=O(n)\).
\end{proof}

\paragraph{Inward snap.}
For an arbitrary cut \(z\) of \(Y\), let \(z^+\) be the least element of
\(G_t\) that is at least \(z\), and let \(z^-\) be the greatest element of
\(G_t\) that is at most \(z\).
If \(z\notin G_t\), it lies inside one greedy cell consisting only of
\(s\)-light characters.  The clipped and unclipped gap masses of this cell
coincide and are below \(2s\).  Thus moving a left endpoint to \(z^+\) or a
right endpoint to \(z^-\) never crosses an \(s\)-heavy character and deletes
characters of total gap cost less than \(2s\).

For a \(Y\)-window \(W=Y[u\mathinner{.\,.}v)\), we define its inward snap by
\[
 \operatorname{snap}_t(W)=
 \begin{cases}
 Y[u^+\mathinner{.\,.}v^-),&u^+<v^-,\\
 Y[v^-\mathinner{.\,.}v^-),&u^+\ge v^-.
 \end{cases}
\]
If \(u^+<v^-\), then \(u\le u^+<v^-\le v\), so the nonempty snap is a
subwindow of \(W\).  The total unclipped gap mass removed from its two ends
is less than \(4s=\eps t/16\).

When \(u^+\ge v^-\), the snap is empty and is anchored at the grid cut \(v^-\).
There are two cases.  If \(u^+=v^-=a\), then
\(u\le a\le v\), so the anchor lies inside the original window.  The two
endpoint portions \(Y[u\mathinner{.\,.}a)\) and \(Y[a\mathinner{.\,.}v)\) each have gap mass below
\(2s\), and together they cover \(W\).  Thus \(M(W)<4s\).

If \(u^+>v^-\), there is no grid cut in \([u,v]\).  Any such cut \(z\)
would satisfy \(u^+\le z\le v^-\).  Consequently, \(v^-\) and \(u^+\)
are consecutive grid cuts and
\[
                         v^-<u\le v<u^+.
\]
Both endpoints lie strictly inside one greedy cell of \(s\)-light
characters, so \(M(W)<2s\).  In this case the chosen anchor \(v^-\)
lies strictly to the left of \(u\).  Thus an empty snap need not have its
anchor inside the original window.  In either empty case we have
\[
                             M(W)<4s.
\]

At one threshold, the map \(z\mapsto z^-\) is nondecreasing.  Therefore
the anchors \(v^-\) of empty snaps preserve the order of the original
windows, including for consecutive empty snaps.  Across thresholds,
the maps use different grids, so the empty anchors need not be ordered.
The analysis in \cref{sec:analysis} combines boxes from different
thresholds, so it omits empty canonical windows from the chain and
accounts for them by deleting their \(X\)-windows and inserting their
optimal \(Y\)-windows.  Nonempty canonical windows lie inside their
optimal \(Y\)-windows and hence remain ordered across thresholds.

\subsection{Tag families}

We need a common family of candidate \(Y\)-windows from which the dense
phase can sample.  We group \(X\)-windows into bins according to their
clipped mass.  A \emph{tag} records two cuts of \(Y\) together with a mass
bin.  The cuts identify the \(Y\)-window, and the bin specifies the class of
\(X\)-windows for which that \(Y\)-window was generated.  Tags from different bins remain
distinct even when they have the same cuts.  Each \(X\)-window considers only
tags from its own bin.

Let \(p_i\) be the start coordinate of \(x_i\) in \(X\), so that
\(M_{t,X}(p_i)\) is the clipped prefix mass before \(x_i\), and let
\(m_i=M_t(x_i)\) be the clipped mass of \(x_i\).  We assign \(x_i\) to the
bin \(b_i=\lfloor m_i/s\rfloor\).
The clipped mass of an \(X\)-window in bin \(b\) lies in
\([bs,(b+1)s)\).  For a candidate \(Y\)-window we allow an extra \(2t\)
on each side of this interval to account for the cost of the optimal
subpath and the inward-snap displacement.  For each occupied bin \(b\)
and start cut \(a\in G_t\) we set
\[
 E_b(a)=\left\{e\in G_t:
 \begin{array}{l}
 e\ge a,\\
 \max\{0,bs-2t\}\le M_{t,Y}(e)-M_{t,Y}(a)\le(b+1)s+2t
 \end{array}
 \right\}\cup\{a\}.
\]
Including the endpoint \(e=a\) provides a tag representing the positioned
empty window.  These endpoints define the global tag family
\[
\cY_t^{\rm tag}
   =\{(b,a,e):b\text{ occupied},\ a\in G_t,\ e\in E_b(a)\}.
\]
For a tag \(\tau=(b_\tau,a_\tau,e_\tau)\in\cY_t^{\rm tag}\) we define its represented
window by \(\operatorname{win}(\tau):=Y[a_\tau\mathinner{.\,.}e_\tau)\).

For each window \(x_i\) we restrict the global family to its bin and to starts
whose clipped prefix values lie in a band around \(M_{t,X}(p_i)\).  For the fixed
positive guess \(\Delta\) its candidate family is
\begin{equation}\label{eq:window-family}
 \cY_t(i)=\{(b_i,a,e):
       a\in G_t,\ |M_{t,Y}(a)-M_{t,X}(p_i)|\le\Delta+2s,\ e\in E_{b_i}(a)\}.
\end{equation}
Membership identifies candidates for comparison.  Acceptance is determined by
the count test and the simplified-distance test in the next subsection.

The global family determines the dense sample size, while each individual
family bounds the candidates considered for one \(X\)-window.  We also
bound the union of candidate families for windows that can share a center.
This lets those windows share an enumeration of candidate tags.

\begin{lemma}\label{lem:family-sizes}
The global family satisfies
\[
                         |\cY_t^{\rm tag}|=O(nd/\eps^2).
\]
For every window \(x_i\),
\[
              |\cY_t(i)|=O\!\left(\frac{1+\Delta/t}{\eps^2}\right).
\]
For every fixed tag \(c=(b,a_c,e_c)\in\cY_t^{\rm tag}\),
\[
 \left|\bigcup_{i:\,c\in\cY_t(i)}\cY_t(i)\right|
       =O\!\left(\frac{1+\Delta/t}{\eps^2}\right).
\]
\end{lemma}
\begin{proof}
The mass inequalities defining \(E_b(a)\) restrict \(M_{t,Y}(e)\) to
an interval of width at most \(4t+s\).  By \cref{lem:grid-density} there
are \(O(1/\eps)\) endpoints for each bin and start even after adding
\(e=a\).  The bound \(0<m_i\le dt\) gives \(O(d/\eps)\) occupied bins.
Combining these bounds with the \(O(n)\) starts proves the global bound.
The prefix band in \cref{eq:window-family} has width \(2\Delta+4s\).
Grid density gives \(O(1+\Delta/(\eps t))\) starts in this band.
Multiplying by the endpoint bound proves the bound for each window \(x_i\).

If \(c\in\cY_t(i)\), then \(b_i=b\) and the prefix band of \(x_i\)
contains \(M_{t,Y}(a_c)\).  Thus \(|M_{t,Y}(a_c)-M_{t,X}(p_i)|\le\Delta+2s\).
For any start \(a\) of a tag in \(\cY_t(i)\),
\[
 |M_{t,Y}(a)-M_{t,Y}(a_c)|\le2\Delta+4s.
\]
Thus the start cuts appearing in the union lie in one interval of clipped
prefix values of width \(4\Delta+8s\), and each has \(O(1/\eps)\) ends under the
same bin.
This proves the bound for the union.
\end{proof}

We construct the global family once per threshold and index it in
lexicographic order.  We store \(M_{t,Y}(e)\) for all grid cuts \(e\) in
a sorted array.  For a fixed bin both bounds on \(M_{t,Y}(e)\) increase
as the start \(a\) advances through \(G_t\).  Two pointers locate these
bounds and each traverses the array at most once.  There are
\(O(1/\eps)\) endpoints to output per start including the empty endpoint.
Thus enumeration takes \(O(n/\eps)\) time per bin and \(O(nd/\eps^2)\)
time over all occupied bins.

\subsection{Canonical windows}

Fix \(\Delta\ge D\).  Snapping the fixed optimal window
\(y_i=Y[z_{i-1}^\star\mathinner{.\,.}z_i^\star)\) gives the \emph{canonical
window} at threshold \(t\).  Its \emph{canonical tag} \(\mu_t(i)=(b_i,a,e)\) is defined by
\[
       \operatorname{win}(\mu_t(i))=Y[a\mathinner{.\,.}e)=\operatorname{snap}_t(y_i).
\]
The canonical tag is \emph{present} when it belongs to \(\cY_t(i)\).
A pair \((i,\tau)\) is \emph{accepted} when the count test passes and
\(\SD_t(x_i,\operatorname{win}(\tau))\le t\).

Let \(\delta_{i,t}\) be the unclipped gap mass of the endpoint portions removed
in a nonempty snap or all of \(M(y_i)\) in an empty snap.  Then
\begin{equation}\label{eq:canonical-removal}
      \delta_{i,t}<4s=O(\eps t).
\end{equation}

Set \(\gamma=1-17\eps/8\).  The margin \((1-\gamma)t\) covers the
snapping loss \(4s\) and the simplification error \(2\eps t\).

\begin{lemma}\label{lem:canonical-present}
If \(D_i\le t\), the canonical tag \(\mu_t(i)\) belongs to \(\cY_t(i)\).
If \(D_i\le\gamma t\), the pair \((i,\mu_t(i))\) is accepted and satisfies
\[
       \SD_t(x_i,\operatorname{snap}_t(y_i))\le t.
\]
\end{lemma}
\begin{proof}
\emph{Presence.}
At the start of the optimal subpath for \(x_i\),
\(|M_{t,Y}(z_{i-1}^\star)-M_{t,X}(p_i)|\le D\le\Delta\) by
\cref{lem:clipped-potential}.  The inward-snap bounds give
\(|M_{t,Y}(a)-M_{t,Y}(z_{i-1}^\star)|<2s\).  This also holds when an empty snap's
anchor lies to the left of \(z_{i-1}^\star\).  The anchor and the original start
belong to the same light cell, whose mass is below \(2s\).  Thus
\(|M_{t,Y}(a)-M_{t,X}(p_i)|<\Delta+2s\), which gives the prefix-band condition in
\cref{eq:window-family}.

The local potential bound gives
\(|m_i-M_t(y_i)|\le D_i\), and inward snapping removes clipped mass below
\(4s\).  The snapped window has endpoints \(a,e\in G_t\) with \(a\le e\), and its
clipped mass differs from \(m_i\) by at most \(D_i+4s\le2t\).
Since \(m_i\in[b_is,(b_i+1)s)\), we obtain
\[
 \max\{0,b_is-2t\}\le M_{t,Y}(e)-M_{t,Y}(a)\le(b_i+1)s+2t.
\]
Thus \(e\in E_{b_i}(a)\).

\emph{Acceptance.}
Suppose \(D_i\le\gamma t\).  Deleting from \(y_i\) the characters removed
by snapping costs \(\delta_{i,t}\), so the triangle inequality gives
\(\ED(x_i,\operatorname{snap}_t(y_i))\le D_i+\delta_{i,t}\).
Applying \cref{lem:simplified-short} and the choice of \(\gamma\) gives
\[
 \SD_t(x_i,\operatorname{snap}_t(y_i))
 \le D_i+\delta_{i,t}+2\eps t<t.
\]
The weighted edit distance is no larger, so \cref{lem:retained-count}
also ensures that the count test passes.
\end{proof}
\section{Dense and sparse phases}\label{sec:threshold}

At each threshold, a shared sample of candidate \(Y\)-windows supplies
centers for the dense phase.  The sparse phase then searches for boxes
for the \(X\)-windows left without a center, using their order to restrict
the search.

Fix a guess \(\Delta>0\), and let \(\gamma\) be the acceptance margin
from \cref{lem:canonical-present}.  We use the threshold sequence
\begin{equation}\label{eq:threshold-ladder}
 t_0=\frac{\eps\Delta}{64q},\qquad
 t_{\ell+1}=(1+\eps)t_\ell,
\end{equation}
through the first index \(L\) with \(t_L\ge\Delta/\gamma\), and call \(\ell\)
the \emph{level} of \(t_\ell\).  This gives
\(L+1=O(\log(q/\eps)/\eps)\), independently of \(\Delta\).

We let \(k\) be an integer parameter with \(1\le k\le n\), set in
\cref{eq:final-parameters}.  An \(X\)-window \(x_i\) is
\emph{\(t\)-dense} if its complete candidate family \(\cY_t(i)\) contains at
least \(k\) tags accepted for \(x_i\), and \emph{\(t\)-sparse} otherwise.

For the analysis, an \(X\)-window \(x_i\) is \emph{good} at threshold
\(t\) if \(D_i\le\gamma t\), and \emph{bad} otherwise.  If
\(\Delta\ge D\), \cref{lem:canonical-present} shows that the canonical
tag \(\mu_t(i)\) is accepted for every good window \(x_i\).

With the constant \(C\) from \cref{sec:preliminaries}, we set
\[
 B_{\rm loc}=C(L+1),\qquad
 \beta=\frac{\gamma}{C(L+1)},\qquad
 B_{\rm sam}=C\log(n/\eps).
\]
Scaling \(B_{\rm loc}\) with \(L+1\) limits the charge from expensive
subpaths to \(O(D/(L+1))\) at each threshold.  The reciprocal choice of
\(\beta\) gives the same bound for sampling failures.  Summing over all
thresholds then gives \(O(D)\) total charge.  The factor \(B_{\rm sam}\)
determines the sample sizes used in the simultaneous sampling guarantee.

\subsection{Dense phase}

A \(t\)-dense \(X\)-window has at least \(k\) accepted tags among the
\(|\cY_t^{\rm tag}|\) global tags.  A uniform sample whose size is
\(|\cY_t^{\rm tag}|/k\) times a sufficiently large logarithmic factor
contains one of these tags with high probability.  We use such a tag
as the window's center.

If \(|\cY_t^{\rm tag}|<k\), no \(X\)-window is
\(t\)-dense, and we set \(s_t^{\mathrm{den}}=0\).  Otherwise we set
\[
 s_t^{\mathrm{den}}=\min\left\{|\cY_t^{\rm tag}|,
       \left\lceil\frac{B_{\rm sam}|\cY_t^{\rm tag}|}{k}\right\rceil\right\}.
\]
Every comparison is preceded by the count test of \cref{sec:simplified}.
We write \(c_i\) for the center assigned to \(x_i\) and \(A_i\) for the
corresponding simplified distance.  \(X\)-windows without a center are
called \emph{residual}, and \(\cX'\subseteq[q]\) denotes their index set.
It depends on \(\Delta\) and \(t\), which we suppress in the notation.

\paragraph{Center-to-candidate comparisons.}
We group the \(X\)-windows assigned the same center tag
\(c=(b,a_c,e_c)\).  All windows in the group use bin \(b\), and their prefix
bands all contain \(M_{t,Y}(a_c)\).  Their union is therefore exactly the interval
\[
 \left[\min_{i:c_i=c}(M_{t,X}(p_i)-\Delta-2s),
       \max_{i:c_i=c}(M_{t,X}(p_i)+\Delta+2s)\right].
\]
Two binary searches in \(M_{t,Y}(G_t)\), followed by the endpoint enumeration
\(E_b(a)\), generate exactly the tags in
\[
                         \bigcup_{i:c_i=c}\cY_t(i).
\]
The bound on unions of candidate families in \cref{lem:family-sizes} shows that
this union contains \(O((1+\Delta/t)/\eps^2)\) tags.
The algorithm enumerates it once per distinct
center and stores the simplified distances in a table \(B_{c,\tau}\).

\begin{algorithm}[H]
\caption{Dense phase at threshold \(t\)}\label{alg:dense-phase}
\small
\begin{algorithmic}[1]
\Require The grid, tag families, and prefix data for \(t\)
\Procedure{DensePhase}{$\Delta,t$}
  \State Initialize empty assignment and center-comparison tables
  \State Sample \(s_t^{\mathrm{den}}\) tags uniformly without replacement from \(\cY_t^{\rm tag}\)
  \For{each \(X\)-window \(x_i\)}
    \For{each sampled tag \(\tau\) in a fixed order}
      \If{\(\tau\in\cY_t(i)\) and the count test passes for \(\operatorname{win}(\tau)\)}
        \State \(A\gets\SD_t(x_i,\operatorname{win}(\tau))\)
        \If{\(A\le t\)}
          \State Store \(c_i\gets\tau\) and \(A_i\gets A\)
          \State \textbf{break} the tag loop
        \EndIf
      \EndIf
    \EndFor
  \EndFor
  \State \(\cX'\gets\{i\in[q]:x_i\text{ has no center}\}\)
  \Statex
  \For{each distinct assigned center \(c=(b,a_c,e_c)\)}
    \For{each tag \(\tau\in\bigcup_{i:c_i=c}\cY_t(i)\)}
      \If{the count test passes for \(\operatorname{win}(\tau)\)}
        \State Store \(B_{c,\tau}\gets\SD_t(\operatorname{win}(c),\operatorname{win}(\tau))\)
      \EndIf
    \EndFor
  \EndFor
  \Statex
  \State \(\mathcal D_t\gets\varnothing\)
  \For{each \(i\in[q]\setminus\cX'\) and each \(\tau\in\cY_t(i)\) with \(B_{c_i,\tau}\) stored}
    \State Append \((i,a_\tau,e_\tau,A_i+B_{c_i,\tau})\) to \(\mathcal D_t\)
  \EndFor
  \State \Return \(\mathcal D_t\) and \(\cX'\)
\EndProcedure
\end{algorithmic}
\end{algorithm}

For each \(X\)-window \(x_i\) with a center, the dense phase outputs a box
for every \(\tau\in\cY_t(i)\) whose window passes the count test.  Its value
\(B_{c_i,\tau}\) was stored, since \(\cY_t(i)\) lies in the union enumerated
for \(c_i\).  The box has label
\[
                         v_t(i,\tau)=A_i+B_{c_i,\tau}.
\]

\begin{lemma}\label{lem:dense-guarantee}
The dense phase satisfies the following properties.
\begin{enumerate}[label=(\roman*),leftmargin=*]
\item Every dense-box label equals the cost of a local edit sequence.
\item Suppose \(\Delta\ge D\).  If \(x_i\) is good and has a center,
then \(\mathcal D_t\) contains the box for its canonical tag, and its label satisfies
\[
 v_t(i,\mu_t(i))
 \le\ED\bigl(x_i,\operatorname{snap}_t(y_i)\bigr)+(2+2\eps)t.
\]
\end{enumerate}
\end{lemma}
\begin{proof}
The edit sequences underlying \(A_i\) and \(B_{c_i,\tau}\) can be
concatenated. The first transforms \(x_i\) into \(\operatorname{win}(c_i)\), and
the second transforms \(\operatorname{win}(c_i)\) into
\(\operatorname{win}(\tau)\).  This proves (i).

For (ii), write \(V=\operatorname{snap}_t(y_i)\).
By \cref{lem:canonical-present}, \(\mu_t(i)\in\cY_t(i)\) and
\(\ED(x_i,V)\le t\).  The canonical window passes the count test, so
\(B_{c_i,\mu_t(i)}=\SD_t(\operatorname{win}(c_i),V)\) is stored and
\(\mathcal D_t\) contains the box for \(\mu_t(i)\).  Center
selection gives \(A_i=\SD_t(x_i,\operatorname{win}(c_i))\le t\).
Applying \cref{cor:dense-composition} with \(U=x_i\) and the center window
\(\operatorname{win}(c_i)\) proves (ii).
\end{proof}

The first pass inspects at most \(s_t^{\mathrm{den}}q\) pairs of an
\(X\)-window and a sampled tag.  Each inspection takes at most
\(O(d^2/\eps)\) time.  Since
\(s_t^{\mathrm{den}}=\softO(nd/(k\eps^2))\) and \(q=O(n/d)\), the first
pass takes \(\softO(n^2d^2/(k\eps^3))\) time at one threshold.  Summing
over \(L+1=\softO(1/\eps)\) thresholds gives
\(\softO(n^2d^2/(k\eps^4))\) time.

There are at most \(s_t^{\mathrm{den}}\) distinct selected centers.
Applying the bound on unions of candidate families
in \cref{lem:family-sizes} to each center bounds the total number of
center-to-candidate comparisons at threshold \(t\) by
\begin{equation}\label{eq:center-target-count}
 \softO\!\left(\frac{nd}{k\eps^4}(1+\Delta/t)\right).
\end{equation}
Since \(\Delta/t_0=64q/\eps\), the geometric sequence satisfies
\begin{equation}\label{eq:threshold-reciprocal-sum}
 \sum_{\ell=0}^{L}(1+\Delta/t_\ell)=O(q/\eps^2).
\end{equation}
Each center-to-candidate comparison takes \(O(d^2/\eps^2)\) time, so
summing \cref{eq:center-target-count} over the thresholds bounds the time of
these comparisons by \(\softO(n^2d^2/(k\eps^8))\).

We produce the dense boxes from the stored comparison results.  A window \(x_i\)
with a center receives at most \(|\cY_t(i)|=O((1+\Delta/t)/\eps^2)\) boxes by
\cref{lem:family-sizes}.  We enumerate \(\cY_t(i)\) by two binary searches in
\(M_{t,Y}(G_t)\) followed by the endpoint enumeration \(E_{b_i}(a)\), and look
up each stored value in \(O(\log n)\) time.  By
\cref{eq:threshold-reciprocal-sum}, the dense phase outputs \(O(q^2/\eps^4)\)
boxes in \(\softO(q^2/\eps^4)=\softO(n^2/(d^2\eps^4))\) time over the
threshold sequence.  Its total time is therefore
\(\softO(n^2d^2/(k\eps^8)+n^2/(d^2\eps^4))\).

Let \(\mathcal E_{\rm dense}\) be the event that every \(t\)-dense \(X\)-window is
assigned a center throughout all guesses and thresholds.  On this event
every residual \(X\)-window is \(t\)-sparse.  A \(t\)-sparse \(X\)-window may
still receive a center, so \(\cX'\) can be smaller than the set of
\(t\)-sparse indices.

\subsection{Sparse phase}

The sparse phase uses the order of the \(X\)-windows to limit the candidate
starts examined for each residual window, since examining every candidate
family in full could be expensive.  The recursion instead samples a few
residual \(X\)-windows from each block and uses their accepted starts to
construct a candidate list for the whole block.

The reason this works is that the canonical starts of \(X\)-windows in one
block are close in clipped prefix mass when the optimal subpath across the
block has sufficiently small cost.  Suppose \(\Delta\ge D\) and the parent
list contains the canonical starts of the child's good residual
\(X\)-windows.  Processing a sampled good \(X\)-window against that list
finds its canonical tag, and a suitable neighborhood of its start then
contains the canonical starts of the other good \(X\)-windows in the child.
The algorithm cannot distinguish
this start from the other accepted starts, so it takes the union of the
neighborhoods around all of them.  Sparsity bounds the number of these
neighborhoods, while grid density bounds the number of candidate starts
they generate.

Recall that \(\cX'\subseteq[q]\) indexes the residual \(X\)-windows
left by the dense phase at threshold \(t\).  We recursively bisect the ordered
window indices, always at the lower middle cut.  Each node is a contiguous
index block \(I\subseteq[q]\) representing \(|I|\) windows.  We write
\(\cX'_I=\cX'\cap I\) for the residual indices in \(I\) and
\(\mathcal L_I\subseteq G_t\) for its sorted list of candidate
starts in \(Y\).  The set \(\cX'\) stays fixed throughout the recursion, including
after boxes are found.  The routines share \(\Delta,t,\cX'\) and the output
list \(\mathcal S_t\).  \emph{ProcessXWindow} appends accepted boxes to this
list and returns their starts for constructing the child lists.

For each child block \(J\) with \(\cX'_J\ne\varnothing\) we draw a fresh
uniform sample \(Q_J\subseteq\cX'_J\) without replacement of size
\begin{equation}\label{eq:sparse-sample-size}
 \min\!\left\{|\cX'_J|,
       \left\lceil\frac{(1+\beta)B_{\rm sam}}{\beta}\right\rceil
       \right\}
\end{equation}
indices.  We run \emph{ProcessXWindow} on each sampled window \(x_i\)
using the parent list \(\mathcal L_I\) restricted to the prefix band of \(x_i\).
Let \(\mathcal A_J\) be the multiset of starts returned by these calls.

The child list consists of neighborhoods around these starts.  Write
\(x_J\) for the concatenation of the \(X\)-windows indexed by \(J\).  The
radius accounts for the clipped mass of \(x_J\) and allows
\(B_{\rm loc}|J|t\) for the intervening edit cost and snapping
displacements.  We set
\[
       \Lambda_J=M_t(x_J)+B_{\rm loc}|J|t,
\]
and define the child list by
\begin{equation}\label{eq:sparse-child-list}
 \mathcal L_J=
 \left\{z\in G_t:
 |M_{t,Y}(z)-M_{t,Y}(a)|\le\Lambda_J
 \text{ for some }a\in\mathcal A_J\right\}.
\end{equation}
\Cref{lem:sparse-localize} specifies when this list contains the canonical
starts of all good residual \(X\)-windows in \(J\).  The neighborhoods
are taken over the global grid \(G_t\) and may contain starts absent
from the parent list.  Two binary searches in \(M_{t,Y}(G_t)\) locate
each neighborhood.

\begin{algorithm}[H]
\caption{Sparse phase at threshold \(t\)}\label{alg:sparse-phase}
\small
\begin{algorithmic}[1]
\Require The grid, tag families, and prefix data for \(t\)
\Procedure{SparsePhase}{$\Delta,t,\cX'$}
  \State \(\mathcal S_t\gets\varnothing\)
  \State \Call{SearchBlock}{$[q],G_t$}
  \State \Return \(\mathcal S_t\)
\EndProcedure
\Statex
\Procedure{SearchBlock}{$I,\mathcal L_I$}
  \If{\(\cX'_I=\varnothing\)}
    \State \Return
  \EndIf
  \If{\(I=\{i\}\)}
    \State \Call{ProcessXWindow}{$i,\mathcal L_I$}
    \State \Return
  \EndIf
  \For{each child \(J\) of \(I\) with \(\cX'_J\ne\varnothing\)}
    \State Sample \(Q_J\subseteq\cX'_J\) as in \cref{eq:sparse-sample-size}
    \State Initialize an empty multiset \(\mathcal A_J\)
    \For{each sampled index \(i\in Q_J\)}
      \State Append the starts returned by \Call{ProcessXWindow}{$i,\mathcal L_I$} to \(\mathcal A_J\)
    \EndFor
    \State Form the sorted list \(\mathcal L_J\) from \cref{eq:sparse-child-list}, without duplicates
    \State \Call{SearchBlock}{$J,\mathcal L_J$}
  \EndFor
\EndProcedure
\Statex
\Procedure{ProcessXWindow}{$i,\mathcal L$}
  \State Initialize an empty multiset \(\mathcal A\) and a counter \(n_{\rm acc}\gets0\)
  \State Find the sublist \(\mathcal L'\subseteq\mathcal L\) in the prefix band of \(x_i\) by two binary searches
  \For{each \(a\in\mathcal L'\) and each \(e\in E_{b_i}(a)\)}
    \If{the count test passes for \(Y[a\mathinner{.\,.}e)\)}
      \State \(v\gets\SD_t(x_i,Y[a\mathinner{.\,.}e))\)
      \If{\(v\le t\)}
        \State Append \((i,a,e,v)\) to \(\mathcal S_t\)
        \State Append \(a\) to \(\mathcal A\)
        \State \(n_{\rm acc}\gets n_{\rm acc}+1\)
        \If{\(n_{\rm acc}=k\)}
          \State Abandon this guess
        \EndIf
      \EndIf
    \EndIf
  \EndFor
  \State \Return \(\mathcal A\)
\EndProcedure
\end{algorithmic}
\end{algorithm}

Each call to \emph{ProcessXWindow} enumerates every candidate tag at most
once.  Its counter therefore counts distinct accepted tags, even when
several have the same start.  The counter resets on each call, and reaching
\(k\) abandons the entire guess.  On \(\mathcal E_{\rm dense}\), this never
happens, because every residual \(X\)-window is \(t\)-sparse.

For the analysis, we call \(x_i\) \emph{recovered} at threshold \(t\) if a
box for its canonical tag \(\mu_t(i)\) belongs to \(\mathcal D_t\) or
\(\mathcal S_t\).

\begin{lemma}
\label{lem:sparse-localize}
Assume \(\Delta\ge D\) and \(\mathcal E_{\rm dense}\), and suppose that
\begin{enumerate}[label=(\roman*),leftmargin=*]
\item the parent list contains the canonical start of every good
residual window \(x_i\) with \(i\in J\);
\item the sample from \(J\) contains the index of a good residual \(X\)-window; and
\item
\[
                 D_J\le\frac{B_{\rm loc}}2|J|t,
\]
where \(D_J\) is the cost of the fixed optimal subpath across \(J\).
\end{enumerate}
Then \(\mathcal L_J\) contains the canonical start of every good
residual window \(x_i\) with \(i\in J\).
\end{lemma}
\begin{proof}
For a window \(x_i\), let \(a_i\in G_t\) be its canonical start.  If
\(x_i,x_{i'}\) are good residual windows with \(i<i'\) and
\(i,i'\in J\), then the optimal subpath between their
starts in \(X\) has cost at most \(D_J\).  By
\cref{lem:clipped-potential},
\[
 |M_{t,Y}(z_{i-1}^\star)-M_{t,Y}(z_{i'-1}^\star)|
 \le M_t(X[p_i\mathinner{.\,.}p_{i'}))+D_J
 \le M_t(x_J)+D_J.
\]
Each canonical start differs from its unsnapped start by less than \(2s\)
in clipped prefix mass.  The triangle inequality therefore gives
\[
 |M_{t,Y}(a_i)-M_{t,Y}(a_{i'})|
 <M_t(x_J)+D_J+4s\le\Lambda_J,
\]
where the last inequality uses assumption (iii) and
\(4s\le(B_{\rm loc}/2)|J|t\).

Choose a sampled index \(i_0\) such that \(x_{i_0}\) is good.
Its canonical start belongs
to the parent list, so \cref{lem:canonical-present} ensures that processing
\(x_{i_0}\) finds its canonical tag.  Hence \(a_{i_0}\in\mathcal A_J\).
The preceding bound places the canonical start of every good
residual \(X\)-window in the neighborhood of \(a_{i_0}\), hence in \(\mathcal L_J\) by
\cref{eq:sparse-child-list}.
\end{proof}

\subsection{Sampling guarantees}\label{sec:coverage}

We condition on the value of a positive rough estimate \(\widehat D\).  This fixes the
guesses and their threshold sequences; it does not assume that the rough
estimate succeeded.
The recovery procedures use fresh randomness independent of the
rough-estimate routine.  All samples are drawn without replacement.

Let \(\mathcal E_{\rm sparse}\) be the event that the following holds for
every sparse sampling call throughout the algorithm.  If at least a
\(\beta\)-fraction of the residual \(X\)-windows in the child are good,
the sample contains the index of a good window.  Set
\(\mathcal E=\mathcal E_{\rm dense}\cap\mathcal E_{\rm sparse}\).

For each dense \(X\)-window or sampled sparse child, the probability of
violating its coverage condition is at most \(e^{-B_{\rm sam}}\).
A union bound over these windows and calls gives the following guarantee.

\begin{lemma}
\label{lem:simultaneous-coverage}
For \(C\) sufficiently large,
\(\mathcal E\) holds with probability at least \(1-1/(2n^{10})\), also
for the adaptively reached sparse children.
\end{lemma}
\begin{proof}
When a sample contains the entire population, the corresponding coverage
condition holds deterministically.  Otherwise, sampling without replacement
has miss probability at most that of sampling with replacement.  A dense
\(X\)-window \(x_i\) has at least \(k\) accepted tags in \(\cY_t(i)\), a subset
of the \(|\cY_t^{\rm tag}|\) global tags.  The probability that the dense sample contains
none of these accepted tags is at most
\[
 \left(1-\frac{k}{|\cY_t^{\rm tag}|}\right)^{s_t^{\mathrm{den}}}
 \le e^{-B_{\rm sam}}.
\]
For a sparse child \(J\) about to be sampled, condition on the history
\(\mathcal H\) before the sample is drawn.  This fixes the child, its
residual index set \(\cX'_J\), and which of the indexed \(X\)-windows are good.  The fresh
sample remains uniform without replacement.  Thus, if at least a
\(\beta\)-fraction of these \(X\)-windows are good, the sample size in
\cref{eq:sparse-sample-size} gives
\[
 \Prb[\text{no good }X\text{-window is sampled}\mid\mathcal H]
 \le e^{-(1+\beta)B_{\rm sam}}
 \le e^{-B_{\rm sam}}.
\]
This bound holds for every possible preceding history in which the
fraction of good residual \(X\)-windows is at least \(\beta\).  For other histories,
the sparse coverage condition imposes no requirement.  Averaging over
the history therefore bounds the probability that this condition fails
at each sampling call, even though earlier samples may determine which
calls are reached.

At each guess and threshold there are at most \(q\) coverage conditions
for dense \(X\)-windows and fewer than \(2q\) sparse sampling calls.
There are \(|\mathcal G(\widehat D)|=\lceil\log_2(64n)\rceil+2\) guesses, each
with \(L+1\) thresholds.  The value of \(L\) is independent of the guess
by \cref{eq:threshold-ladder}.  Conditional on \(\widehat D\), a union bound over
the dense \(X\)-windows and the sparse calls in execution order gives
\[
 \Prb[\mathcal E^c\mid \widehat D]
 \le 3q\,|\mathcal G(\widehat D)|(L+1)e^{-B_{\rm sam}}
 \le \frac{1}{2n^{10}}.
\]
The last inequality holds for all \(n\ge2\) when \(C\) is sufficiently
large, since \(q\le n\), \(|\mathcal G(\widehat D)|=O(\log n)\), and
\(L+1=O(\eps^{-1}\log(n/\eps))\), so the factor multiplying
\(e^{-B_{\rm sam}}\le(\eps/n)^C\) is at most \((n/\eps)^{O(1)}\).
Independence between these coverage events is not needed.  Since the
bound is uniform in \(\widehat D\), it also holds without conditioning on the
rough estimate.
\end{proof}

\subsection{Sparse analysis}\label{sec:sparse-work}

\Cref{lem:sparse-localize} gives the inductive step of the sparse analysis.
At a node \(I\), the invariant is that \(\mathcal L_I\) contains the
canonical start of every good residual \(X\)-window indexed by \(I\).
It holds at the root. If it holds at \(I\), it propagates to a child \(J\)
unless condition (iii) of \cref{lem:sparse-localize} fails or no good
\(X\)-window is sampled from the child.  A child is a \emph{cost-failure node} when
condition (iii) fails and a \emph{sampling-failure node} when that condition
holds but no good \(X\)-window is sampled.
A sampling-failure node can still occur on \(\mathcal E_{\rm sparse}\) when
fewer than a \(\beta\)-fraction of its residual \(X\)-windows are good.

These failure nodes are used only in the analysis. The algorithm cannot
inspect \(D_J\) or determine whether an \(X\)-window is good. In the proof,
we truncate each branch at its first failure node.  No first failure node
is a descendant of another.  Any branch considered in the analysis with no failure
reaches a leaf, where the invariant and \cref{lem:canonical-present} yield a
box for the residual \(X\)-window's canonical tag.

The next lemma bounds the total charge of \(t\) per missed good \(X\)-window
at one threshold.  Each such window is charged to the first failure node
on its branch.

\Needspace{8\baselineskip}
\begin{lemma}
\label{lem:first-failure-charge}
Fix a tight guess and assume \(\mathcal E\).  Then, at every
threshold \(t\),
\[
 \sum_{\substack{i:\,x_i\text{ good at }t\\
                  x_i\text{ not recovered at }t}}t
 =O\!\left(\frac{D}{L+1}\right).
\]
Moreover, every \(X\)-window is recovered at the terminal threshold.
\end{lemma}
\begin{proof}
The optimal subpath across a recursion node \(J\) is the concatenation
of the subpaths \(\pi_i\) for \(i\in J\).  Hence
\[
                              D_J=\sum_{i\in J}D_i.
\]
Because no first failure node is a descendant of another, their index
blocks are disjoint, so their optimal subpaths share at most endpoints.  Consequently, for every
collection \(\mathcal F\) of first failure nodes,
\[
                              \sum_{J\in\mathcal F}D_J\le D.
\]

Every unrecovered good \(X\)-window is residual by \cref{lem:dense-guarantee}.
Assign its charge of \(t\) to the first failure node on its root-to-leaf path.

At a first failure node \(J\) of the cost type, at most \(|J|\) windows are charged, and
condition (iii) of \cref{lem:sparse-localize} fails.  Thus
\[
                         |J|t<\frac{2D_J}{B_{\rm loc}}.
\]
Since these subpath costs sum to at most \(D\), these nodes
contribute at most \(2D/[C(L+1)]\).

At a first failure node \(J\) of the sampling type, fewer than a \(\beta\)-fraction of the
residual \(X\)-windows are good, by \(\mathcal E_{\rm sparse}\).
These counts concern \(\cX'_J\), which may be smaller than the full block
\(J\).  Thus the number of good residual windows is less than
\(\beta/(1-\beta)\) times the number of bad ones.  Since each bad window
\(x_i\) has \(D_i>\gamma t\), the charge to \(J\) is at most
\[
 \frac{\beta}{\gamma(1-\beta)}
 \sum_{\substack{i\in\cX'_J\\D_i>\gamma t}}D_i.
\]
The blocks are disjoint, so these sums of \(D_i\) total at most \(D\).
The choice of \(\beta\) and the inequality \(\beta\le1/8\), which holds
because \(C\ge8\), bound the total charge to these nodes by
\(2D/[C(L+1)]\).
Combining the two bounds proves the missed-window bound.

At the terminal threshold \(t_L\),
\(D_i\le D\le\Delta\le\gamma t_L\) for every window \(x_i\).  Moreover, for every
nonempty child \(J\),
\[
 D_J\le D\le\Delta\le\gamma t_L
      <\frac{B_{\rm loc}}2|J|t_L.
\]
The displayed bound rules out cost-failure nodes.  Every residual \(X\)-window is
good at \(t_L\), so every nonempty child sample contains the index of
such a window.
Hence sampling-failure nodes are also impossible.  The invariant therefore
holds at every residual leaf and yields a box for each residual \(X\)-window's
canonical tag.
For every \(X\)-window with a center, \cref{lem:dense-guarantee} yields its
canonical box as well.
\end{proof}

We next bound the candidate lists, the time needed to examine them, and the
number of sparse boxes.  These bounds hold for every execution up to
completion or abandonment of the guess.  The stopping rule bounds each
\emph{ProcessXWindow} call by \(k\) accepted tags; its running time also depends
on the candidate starts examined.

\paragraph{Candidate lists.}
By the choices of \(\beta\) and \(B_{\rm sam}\), each child samples
\(O((L+1)\log(n/\eps))=\softO(1/\eps)\) indices.  Each completed sampled
call returns fewer than \(k\) starts, counted with multiplicity, so
\(|\mathcal A_J|=\softO(k/\eps)\).  One returned start contributes a
neighborhood of clipped prefix width \(2\Lambda_J\).  By
\cref{lem:grid-density}, this neighborhood contains
\[
 O\!\left(1+\frac{\Lambda_J}{\eps t}\right)
 =O\!\left(1+\frac{M_t(x_J)}{\eps t}+\frac{B_{\rm loc}|J|}{\eps}\right)
\]
grid cuts.  Consequently, the number of entries generated to form
\(\mathcal L_J\), before sorting and deduplication, is at most
\[
 \softO\!\left(
 \frac{k}{\eps^2}\left(1+\frac{M_t(x_J)}{t}+B_{\rm loc}|J|\right)\right).
\]
This estimate counts a cut again whenever it belongs to another
neighborhood.  Deduplication can only reduce the resulting list size.

At any fixed depth, the child blocks are disjoint, so the strings \(x_J\)
are disjoint substrings of \(X\) and
\[
 \sum_J M_t(x_J)\le M_t(X)\le nt,\qquad
 \sum_J |J|\le q,\qquad
 |\{J\}|\le q.
\]
Summing the preceding entry bound over these blocks gives
\[
 \softO\!\left(\frac{k}{\eps^2}(q+n+B_{\rm loc}q)\right)
 =\softO(nk/\eps^3),
\]
using \(q\le n\) and \(B_{\rm loc}=\softO(1/\eps)\).  Thus both the
deduplicated lists and the entries generated while constructing them
have total size \(\softO(nk/\eps^3)\) at each nonroot depth.

\paragraph{Examinations.}
In \emph{ProcessXWindow}, we count each candidate start considered for a
residual \(X\)-window as one examination.  The binary searches
first restrict the supplied list to its prefix band, so a call using
\(\mathcal L_I\) examines at most \(|\mathcal L_I|\) starts.  This
counts examinations whose endpoint candidates are all rejected as well
as those that produce boxes.

For an internal node \(I\), its list is used to process the sampled \(X\)-windows
from its at most two children.  Each child samples
\(\softO(1/\eps)\) indices, so these calls examine each entry of
\(\mathcal L_I\) at most \(\softO(1/\eps)\) times.  At a fixed
nonroot depth, their total number of examinations is therefore at most
\[
 \softO(1/\eps)\sum_{I\text{ at this depth}}|\mathcal L_I|
 =\softO(nk/\eps^4).
\]
A residual leaf makes one additional call using its own list.  At each
depth the leaf lists are included in the same sum, so these final calls
also satisfy the bound.

The root list \(G_t\) is used in \(\softO(1/\eps)\) calls to
\emph{ProcessXWindow}.  In each call, the prefix-band
restriction leaves \(O(1+\Delta/(\eps t))\) starts by grid density.
Since \(t\ge t_0\) and \(\Delta/t_0=64q/\eps\), the root contributes
\(\softO(q/\eps^3)\) examinations per threshold.  There are
\(O(\log(q+1))\) depths and \(L+1=\softO(1/\eps)\) thresholds.
Summing over them gives \(\softO(nk/\eps^5)\) examinations
throughout the threshold sequence.

\paragraph{Sparse boxes.}
At each depth, an index \(i\) belongs to one block and can be sampled at most
once because sampling is without replacement.  The window \(x_i\) is
also processed once at its leaf if the recursion reaches that leaf.
Thus each \(X\)-window participates
in at most \(\lceil\log_2 q\rceil+1\) calls per threshold.  With at most
\(k\) boxes stored per call, the total size of the lists \(\mathcal S_t\)
over all thresholds is at most
\[
 qk\bigl(\lceil\log_2 q\rceil+1\bigr)(L+1)
 =\softO(qk/\eps).
\]
This includes repeated boxes found in different calls.

\paragraph{Running time.}
For each examined start, the endpoint family has \(O(1/\eps)\) members.
The count test takes constant time per endpoint, and each comparison
that passes it takes \(O(d^2/\eps)\) time.  Thus the examinations take
\(\softO(nkd^2/\eps^7)\) time in total.  Sampling and the binary searches require
only logarithmic overhead per call or returned start.  There are
\(\softO(q/\eps)\) calls across the threshold sequence and at most
\(\softO(qk/\eps)\) returned starts.  Generating, sorting, and
deduplicating all candidate lists takes
\(\softO(nk/\eps^4)\) additional time by the entry bounds above.
These additional costs are dominated by the running-time bound for processing
candidate starts.  The sparse phase therefore takes \(\softO(nkd^2/\eps^7)\) time
over the threshold sequence.

\subsection{Threshold sequence guarantee}

Grid and tag preparation take \(O(nd/\eps^2)\) time per threshold by
\cref{lem:family-sizes} and \(\softO(nd/\eps^3)\) time over the sequence.
Together with the dense and sparse phases, this gives the following guarantee.

\begin{samepage}
\begin{theorem}
\label{thm:threshold-ladder}
Fix a guess \(\Delta>0\).  Every box produced during the threshold
sequence has a label equal to the cost of a local edit sequence.  Processing
these thresholds takes
\[
                 \softO\!\left(
                 \frac{nd}{\eps^3}+\frac{n^2d^2}{k\eps^8}
                 +\frac{nkd^2}{\eps^7}+\frac{n^2}{d^2\eps^4}\right)
\]
time before either completing or abandoning the guess.  A completed
sequence produces \(P_{\rm box}=\softO(q^2/\eps^4+qk/\eps)\) boxes.  If
\(D\le\Delta<2D\) and \(\mathcal E\) holds, it completes and satisfies the
following properties.
\begin{enumerate}[label=(\roman*),leftmargin=*]
\item at every threshold \(t\), every good \(x_i\) with a center has a
dense box for its canonical tag, with the label bound in
\cref{lem:dense-guarantee};
\item every sparse box for a canonical tag at threshold \(t_\ell\) satisfies
\[
 \SD_{t_\ell}\bigl(x_i,\operatorname{snap}_{t_\ell}(y_i)\bigr)
 \le\ED\bigl(x_i,\operatorname{snap}_{t_\ell}(y_i)\bigr)+2\eps t_\ell;
\]
\item the missed-window charge satisfies the bound in
\cref{lem:first-failure-charge}; and
\item every \(X\)-window is recovered at the terminal threshold.
\end{enumerate}
\end{theorem}
\end{samepage}

The samples and recursion lists are local to one threshold.  The lists
\(\mathcal D_t\) and \(\mathcal S_t\) from all thresholds are collected in
\(\mathcal B\) and passed to the chunked DP in \cref{sec:aggregation}.
\section{Proof of the main theorem}\label{sec:proof}

\subsection{Approximation guarantee}\label{sec:analysis}

Fix a tight guess and assume the event \(\mathcal E\) of
\cref{sec:coverage}.  For each \(X\)-window \(x_i\), let \(\sigma_i\) be
the first level at which it is good, and let \(j_i\) be the first level at
or after \(\sigma_i\) at which it is recovered.  Both levels exist because
every \(X\)-window is good and recovered at \(t_L\)
(\cref{thm:threshold-ladder}).

For every \(\sigma_i\le\ell<j_i\), window \(x_i\) is good but unrecovered at
threshold \(t_\ell\).  By \cref{lem:first-failure-charge}, these windows
have total charge \(O(D/(L+1))\) at each level.  Summing over levels gives
\[
 \sum_i\sum_{\ell=\sigma_i}^{j_i-1}t_\ell
 =\sum_{\ell=0}^{L}t_\ell
   \bigl|\{i:\sigma_i\le\ell<j_i\}\bigr|=O(D).
\]

If \(\sigma_i>0\), minimality gives
\(t_{\sigma_i}<(1+\eps)D_i/\gamma\).  The windows with \(\sigma_i=0\)
contribute at most \(qt_0=O(\eps D)\).  Thus
\[
 \sum_it_{\sigma_i}\le\frac{1+\eps}{\gamma}D+qt_0
 =(1+O(\eps))D.
\]
Since consecutive thresholds differ by \(\eps t_\ell\), the total
increase from late recoveries is
\begin{equation}\label{eq:delay-identity}
 \sum_i(t_{j_i}-t_{\sigma_i})
 =\eps\sum_i\sum_{\ell=\sigma_i}^{j_i-1}t_\ell=O(\eps D).
\end{equation}
Consequently,
\begin{equation}\label{eq:first-scale-sum}
 \sum_it_{j_i}\le(1+O(\eps))D.
\end{equation}

\begin{lemma}
\label{lem:tight-comparison-chain}
Fix a tight guess and assume \(\mathcal E\).
Then the box collection contains a
monotone chain of cost at most \((3+O(\eps))D\), where the hidden
constant is absolute.
\end{lemma}
\begin{proof}
For every \(x_i\) whose canonical window at level \(j_i\) is nonempty,
include a recovered box in \(\mathcal C\), in increasing order of \(i\).
The optimal windows \(y_i\) are ordered, and each included \(Y\)-window
lies inside \(y_i\).  Thus these boxes form a monotone chain even when the
levels \(j_i\) differ.  We omit empty canonical windows, whose anchors
need not remain ordered across levels.

Let \(\mathcal N\) be the set of indices \(i\) whose canonical window is nonempty.  For
\(i\in \mathcal N\), write \(\lambda_i\) for the recovered box's label and
\(\delta_i=\delta_{i,t_{j_i}}\) for the unclipped gap mass removed from
\(y_i\).  Since the optimal \(Y\)-windows partition \(Y\), the chain's
insertion terms charge \(\delta_i\) for \(i\in \mathcal N\) and \(M(y_i)\)
otherwise.  These portions are disjoint, so the exact chain cost is
\[
 \operatorname{cost}(\mathcal C)
 =\sum_{i\in \mathcal N}(\lambda_i+\delta_i)
  +\sum_{i\notin \mathcal N}\bigl(M(x_i)+M(y_i)\bigr).
\]

For \(i\in \mathcal N\), \cref{eq:canonical-removal} gives
\(\delta_i=O(\eps t_{j_i})\).  By \cref{thm:threshold-ladder}, the additive
error is at most \(2\eps t_{j_i}\) for sparse boxes and
\((2+2\eps)t_{j_i}\) for dense boxes.  Thus both satisfy
\begin{align*}
 \lambda_i
 &\le\ED\bigl(x_i,\operatorname{snap}_{t_{j_i}}(y_i)\bigr)
       +(2+2\eps)t_{j_i}\\
 &\le D_i+\delta_i+(2+2\eps)t_{j_i}.
\end{align*}
One copy of \(\delta_i\) comes from replacing \(y_i\) by its canonical
window in the local edit sequence, and the other is the chain's insertion cost.
Hence \(x_i\) contributes at most
\[
 D_i+2\delta_i+(2+2\eps)t_{j_i}
 \le D_i+(2+O(\eps))t_{j_i}.
\]

For \(i\notin \mathcal N\), the canonical window is empty, so
\cref{eq:canonical-removal} gives \(M(y_i)=O(\eps t_{j_i})\).  The triangle
inequality through \(y_i\) gives \(M(x_i)\le D_i+M(y_i)\).  Therefore,
\[
 M(x_i)+M(y_i)
 \le D_i+2M(y_i)=D_i+O(\eps t_{j_i}).
\]
Summing these contributions and using \(\sum_iD_i=D\) together with
\cref{eq:first-scale-sum} gives
\begin{align*}
 \operatorname{cost}(\mathcal C)
 &\le D+(2+O(\eps))\sum_it_{j_i}\\
 &\le(3+O(\eps))D.\qedhere
\end{align*}
\end{proof}

\subsection{Completing the proof}\label{sec:total}

\begin{proof}[Proof of \cref{thm:main}]
By \cref{thm:threshold-ladder,lem:chain-sound}, the minimum returned over
completed guesses and the all-gap candidate is at least \(D\).  The case
\(D=0\) is handled by \emph{RoughUpperBound}.

Suppose \(D>0\).  By \cref{thm:rough-guide,lem:simultaneous-coverage}
and a union bound, the rough estimate succeeds and \(\mathcal E\) holds
with probability at least \(1-n^{-10}\).  On this event, a tight guess
exists and completes by \cref{thm:threshold-ladder}.  For that guess,
\cref{lem:tight-comparison-chain,lem:chain-dp} bound the returned value
by \((3+O(\eps))D\).

For every guess and every choice of randomness,
\cref{thm:threshold-ladder} bounds the time spent before the threshold
sequence completes or the guess is abandoned.  For a completed sequence, the chunked DP evaluates the
\(P_{\rm box}\) boxes of \cref{thm:threshold-ladder} in
\(O(n+P_{\rm box}\log n)\) time by \cref{lem:chain-dp}.  Including \emph{RoughUpperBound} and all \(O(\log n)\)
guesses gives the worst-case running time
\begin{equation}\label{eq:master-resource}
 \softO\!\left(
 n+\frac{nd}{\eps^3}+\frac{n^2d^2}{k\eps^8}
 +\frac{nkd^2}{\eps^7}+\frac{n^2}{d^2\eps^4}
 \right).
\end{equation}
The cost of creating and processing the dense boxes contributes the last
term.  Processing the sparse boxes is dominated by the sparse-phase bound.
Choose
\begin{equation}\label{eq:final-parameters}
 d=\lceil n^{1/8}\rceil,\qquad k=\lceil n^{1/2}\rceil.
\end{equation}
The running time is then \(\softO(n^{7/4}/\eps^8)\), independently of the
ratio between positive edit costs.

It remains to rescale the accuracy.  The bounds above give ratio
\(3+O(\eps)\) with an absolute hidden constant.  Given the accuracy
\(0<\eps\le1\) requested in \cref{thm:main}, we run the algorithm with accuracy
parameter \(\eps'=\eps/C\), which lies in \((0,1/C]\).  Since
\(C\) is at least the hidden constant, the ratio is at most
\(3+C\eps'\le3+\eps\), and the running time is
\(\softO(n^{7/4}/(\eps')^8)=\softO(n^{7/4}/\eps^8)\).
\end{proof}

\section*{AI Disclosure.}
We used OpenAI's GPT-6 Astra to assist in developing and checking the proof
arguments, probability bounds, and parameter choices, and in drafting and
revising the manuscript.  We used Anthropic's Claude Opus 5 to assist in
revising the exposition and notation and in checking the proofs.  The authors
take full responsibility for the final content.

\bibliographystyle{alpha}
\bibliography{ref.bib}

\end{document}